\documentclass[aos,amsmath,amsthm,preprint]{imsart}
\setattribute{journal}{name}{}
\usepackage{graphicx}
\usepackage{amsfonts,amssymb,mathrsfs}
\usepackage{bbm}
\usepackage{subfig}
\usepackage{comment}
\usepackage[square,compress,sort,numbers]{natbib}
\usepackage{pxfonts}
\usepackage{eulervm}
\usepackage{enumerate}
\usepackage{bm}
\usepackage{url}
\usepackage{algorithm,algorithmic}
\usepackage[colorlinks=true,pagebackref,linkcolor=magenta]{hyperref}
\usepackage{parskip}
\newtheorem{theorem}{Theorem}
\newtheorem{lemma}[theorem]{Lemma}
\newtheorem{corollary}{Corollary}

\theoremstyle{definition}
\newtheorem{definition}{Definition}
\newtheorem*{remark}{Remark}
\numberwithin{equation}{section}
\mathchardef\ordinarycolon\mathcode`\:
\mathcode`\:=\string"8000
\begingroup \catcode`\:=\active
  \gdef:{\mathrel{\mathop\ordinarycolon}}
\endgroup

\newcommand{\argmax}{\operatornamewithlimits{argmax}}

\renewcommand{\tilde}{\widetilde}
\renewcommand{\Pr}{\mathbb{P}}

\newcommand{\bU}{\mathbf{U}}
\newcommand{\bX}{\mathbf{X}}
\newcommand{\bLam}{\bm{\Lambda}}
\renewcommand{\Pr}{\mathbb{P}}

\newcommand{\bA}{\mathbf{A}}
\newcommand{\bB}{\mathbf{B}}

\newcommand{\bM}{\mathbf{M}}

\newcommand{\bP}{\mathbf{P}}

\makeatletter
\def\thm@space@setup{%
 \thm@preskip=\parskip \thm@postskip=0pt
}
\makeatother

\begin{document}
\begin{frontmatter}
\title{Asymptotically efficient estimators for stochastic blockmodels: the naive MLE, the rank-constrained MLE, and the spectral}
\runtitle{The naive, the rank-constrained, and the spectral}
\begin{aug} 
\author{\fnms{Minh}
\snm{Tang}\ead[label=e1]{mtang10@jhu.edu}}
\and
\author{\fnms{Joshua}
\snm{Cape}\ead[label=e2]{joshua.cape@jhu.edu}}
\and
\author{\fnms{Carey E.}
\snm{Priebe}\corref{}\ead[label=e3]{cep@jhu.edu}}
 \thankstext{t1}{Supported in part by Johns Hopkins
 University Human Language Technology Center of Excellence (JHU HLT
 COE), the XDATA, SIMPLEX, and D3M programs 
 of the Defense Advanced Research Projects
 Agency (DARPA) administered through contract FA8750-12-2-0303, contract N66001-15-C-4041, and contract FA8750-17-2-0112, respectively.}
 \runauthor{Tang, Cape and Priebe}
 \affiliation{Johns Hopkins University}
\address{Department of Applied Mathematics and Statistics,\\Johns
  Hopkins University,\\3400 N. Charles St,\\Baltimore, MD 21218, USA.\\
\printead{e1}\\
\phantom{Email: }
\printead*{e2} \\
\phantom{Email: }
\printead*{e3}} 
\end{aug}

\begin{abstract} 
We establish asymptotic normality results for estimation of the block probability matrix $\mathbf{B}$ in stochastic blockmodel graphs using spectral embedding when the average degrees grows at the rate of $\omega(\sqrt{n})$ in $n$, the number of vertices. As a corollary, we show that when $\mathbf{B}$ is of full-rank,  estimates of $\mathbf{B}$ obtained from spectral embedding are asymptotically efficient. When $\mathbf{B}$ is singular the estimates obtained from spectral embedding can have smaller mean square error than those obtained from maximizing the log-likelihood under no rank assumption, and furthermore, can be almost as efficient as the true MLE that assume known $\mathrm{rk}(\mathbf{B})$. Our results indicate, in the context of stochastic blockmodel graphs, that spectral embedding is not just computationally tractable, but that the resulting estimates are also admissible, even when compared to the purportedly optimal but computationally intractable maximum likelihood estimation under no rank assumption. 
\end{abstract}

\end{frontmatter}

\section{Introduction}
\label{sec:intro}
Statistical inference on graphs is a burgeoning field of research in machine learning and statistics, with numerous applications to social network, neuroscience, etc. Many of the graphs in application domains are large and complex but nevertheless are believed to be composed of multiple smaller-scale communities. Thus, an essential task in graph inference is detecting/identifying local (sub)communities. 
The resulting problem of community detection on graphs is well-studied (see the survey \cite{Fortunato2010Community}), 
with many available techniques including those based on maximizing modularity and likelihood \cite{Bickel2009,Newman2004,Snijders1997Estimation,Celisse}, random walks \cite{pons05:_comput,rosvall08:_maps}, spectral clustering \cite{mcsherry,rohe2011spectral,chaudhuri12:_spect,sussman12,rinaldo_2013,von2007tutorial}, and semidefinite programming \cite{hajek,abbe}. It is widely known that under suitable models --- such as the popular stochastic blockmodel and its variants \cite{holland,karrer2011stochastic,lyzinski15_HSBM} --- one can consistently recover the underlying communities as the number of observed nodes increases, and furthermore, there exist deep and beautiful phase transitions phenomena with respect to the statistical and computational limits for recovery.

Another important question in graph inference is, subsequent to community detection, that of characterizing the nature and/or structure of these communities. One of the simplest and possibly most essential example is determining the (community specific) tendencies/probabilities for nodes to form link within and between communities. Consistent recovery of the underlying communities yields a straightforward and universally employed procedure for consistent estimation of these probabilities, namely averaging the number of edges within a community and/or between communities. This procedure, when intepreted in the context of estimating the parameters $\bm{\theta} = (\theta_1, \theta_2, \dots, \theta_r)$ of a collection of independent binomially distributed random variables $\{X_1,X_2, \dots, X_r\}$ with $X_i \sim \mathrm{Bin}(n_i, \theta_i)$, corresponds to maximum likelihood estimation with no restrictive assumption on 
the $\{\theta_i\}$. However, in the context of graphs and communities, there are often natural relationships among the communities; hence a graph with $K$ communities need not require $K(K+1)/2$ parameters to describe the within and between community connection probabilities. The above procedure is therefore potentially sub-optimal. 

Motivated by the above observation, our paper studies the asymptotic properties of three different estimators for $\mathbf{B}$, the matrix of edge probabilities among communities, of a stochastic blockmodel graph. Two estimators are based on maximum likelihood methods and the remaining estimator is based on spectral embedding. We show that, given an observed graph with adjacency matrix $\mathbf{A}$, the most commonly used estimator --- the MLE under no rank assumption on $\mathbf{B}$ --- is sub-optimal when $\mathbf{B}$ is not invertible. Moreover, when $\mathbf{B}$ is singular, the estimator based on spectral embedding $\mathbf{A}$ is often times better (smaller mean squared error) than the MLE under no rank assumption, and is almost as efficient as the asymptotically (first-order) efficient MLE whose parametrization depend on $\mathrm{rk}(\mathbf{B})$. Finally, when $\mathbf{B}$ is invertible, the three estimators are asymptotically first-order efficient.

\subsection{Background}
We now formalize the setting considered in this paper. We begin by recalling the notion of stochastic blockmodel graphs due to \cite{holland}. Stochastic blockmodel graphs and its variants, such as degree-corrected blockmodels and mixed membership models \cite{karrer2011stochastic,Airoldi2008} are the most popular models for graphs with intrinsic community structure. In addition, they are widely used as building blocks for constructing approximations (see e.g., \cite{gao,wolfe13:_nonpar,airoldi13:_stoch,klopp}) of the more general latent position graphs or graphons models \cite{Hoff2002,lovasz12:_large}. 
\label{subsec:background}
\begin{definition}
  \label{def:SBM}
  Let $K \geq 1$ be a positive integer and let $\bm{\pi} \in \mathcal{S}_{K-1}$ be a non-negative vector in $\mathbb{R}^{K}$
  with $\sum_{k} \pi_k = 1$; here $\mathcal{S}_{K-1}$ denote the $K-1$ dimensional simplex in $\mathbb{R}^{K}$. Let $\mathbf{B} \in [0,1]^{K \times K}$
  be symmetric. We say that $(\mathbf{A}, \bm{\tau}) \sim
  \mathrm{SBM}(\mathbf{B}, \bm{\pi})$ with sparsity factor $\rho$ if the following hold. First
  $\bm{\tau} = (\tau_1, \dots, \tau_n)$ where $\tau_i$ are
  i.i.d. with $\Pr[\tau_i = k] = \pi_k$. Then
  $\mathbf{A} \in \{0,1\}^{n \times n}$ is a symmetric matrix such
  that, conditioned on $\bm{\tau}$, for all $i \leq j$ the $\mathbf{A}_{ij}$
  are independent Bernoulli random variables with $\mathbb{E}[\mathbf{A}_{ij}] = \rho \mathbf{B}_{\tau_i, \tau_j}$. We write $\mathbf{A} \sim \mathrm{SBM}(\mathbf{B}, \bm{\pi})$ when only $\mathbf{A}$ is observed, i.e., $\bm{\tau}$ is integrated out from $(\mathbf{A}, \bm{\tau})$.
\end{definition}
For $(\mathbf{A}, \bm{\tau}) \sim \mathrm{SBM}(\mathbf{B}, \bm{\pi})$ or $\mathbf{A} \sim \mathrm{SBM}(\mathbf{B}, \bm{\pi})$
with $\mathbf{A}$ having $n$ vertices and (known) sparsity factor $\rho$, 
the likelihood of $(\mathbf{A}, \tau)$ and $\mathbf{A}$ are, respectively
\begin{gather} L(\mathbf{A}, \tau; \mathbf{B}, \bm{\pi}) = \Bigl(\prod_{i=1}^{n} \pi_{\tau_i} \Bigr) \Bigl(\prod_{i \leq j} (\rho \mathbf{B}_{\tau_i, \tau_j})^{\mathbf{A}_{ij}} (1 - \rho \mathbf{B}_{\tau_i,\tau_j})^{1 - \mathbf{A}_{ij}} \Bigr), \\
 L(\mathbf{A}; \mathbf{B}, \bm{\pi}) = \sum_{\bm{\tau} \in [K]^{n}} L(\mathbf{A}, \bm{\tau}; \mathbf{B}, \bm{\pi}). \end{gather}
When we observed $\mathbf{A} \sim \mathrm{SBM}(\mathbf{B}, \bm{\pi})$ with $\mathbf{B} \in [0,1]^{K}$ where $\rho$ and $K$ are assumed known, 
a maximum likelihood estimate of $\bm{\pi}$ and $\mathbf{B}$ is given by
\begin{equation}
\label{eq:MLE_naive}
 (\hat{\mathbf{B}}^{(N)}, \hat{\bm{\pi}}) = \argmax_{\mathbf{B} \in [0,1]^{K \times K}, \,\, \bm{\pi} \in \mathcal{S}_{K-1}} 
L(\mathbf{A}; \mathbf{B}, \bm{\pi}).
\end{equation}
The maximum likelihood estimate (MLE) $\hat{\mathbf{B}}^{(N)}$ in Eq.~\eqref{eq:MLE_naive} 
requires estimation of the $K(K+1)/2$ entries of $\mathbf{B}$ and is generally intractable as it requires marginalization over the latent vertex-to-block assignment vector $\bm{\tau}$. Another parametrization of $\mathbf{B}$ is via the eigendecomposition $\mathbf{B} = \mathbf{V} \mathbf{D} \mathbf{V}^{\top}$ where $d = \mathrm{rk}(\mathbf{B})$, $\mathbf{V}$ is a $K \times d$ matrix with $\mathbf{V}^{\top} \mathbf{V} = \mathbf{I}$ and $\mathbf{D}$ is diagonal.
This parametrization results in the estimation of $d(2K - d+1)/2 \leq K(K+1)/2$ parameters, with $Kd - d(d+1)/2$ parameters being estimated for $\mathbf{V}$ (as an element of the Stiefel manifold of orthonormal $d$ frames in $\mathbb{R}^{K}$) and $d$ parameters estimated for $\mathbf{D}$. Therefore, when $d = \mathrm{rk}(\mathbf{B})$ (in addition to $\rho$ and $K$) is also assumed known, another MLE of $\mathbf{B}$ and $\bm{\pi}$ is given by
\begin{equation}
\label{eq:MLE_rank}
 (\hat{\mathbf{B}}^{(M)}, \hat{\bm{\pi}}) = \argmax_{\mathbf{B} \in [0,1]^{K \times K},\,\, \mathrm{rk}(\mathbf{B}) = d, \,\, \bm{\pi} \in \mathcal{S}_{K-1}}
L(\mathbf{A}; \mathbf{B}, \bm{\pi}).
\end{equation}
The MLE parametrization in Eq.~\eqref{eq:MLE_naive} is the one that is universally used, see e.g., \cite{Choi2010,Celisse,Snijders1997Estimation,bickel_asymptotic_normality}; variants of MLE estimation such as maximization of the profile likelihood \cite{Bickel2009} or variational inference \cite{daudin} are also based solely on approximating the MLE in Eq.~\eqref{eq:MLE_naive}. In contrast, the MLE parametrization used in Eq.~\eqref{eq:MLE_rank} has, to the best of our knowledge, never been considered heretofore in the literature. We shall refer to $\hat{\mathbf{B}}^{(N)}$ and $\hat{\mathbf{B}}^{(M)}$ as the naive MLE and the true (rank-constrained) MLE, respectively. 

The estimator $\hat{\mathbf{B}}^{(N)}$ is asymptotically normal around $\mathbf{B}$; in particular Lemma~1 (and its proof) in \cite{bickel_asymptotic_normality} states that
\begin{theorem}[\cite{bickel_asymptotic_normality}]
\label{thm:bickel_choi}
Let $\mathbf{A}_n \sim \mathrm{SBM}(\mathbf{B}, \bm{\pi})$ for $n \geq 1$ be a sequence of stochastic blockmodel graphs with sparsity factors $\rho_n$. 
Let $\hat{\mathbf{B}}^{(N)}$ be the MLE of $\mathbf{B}$ obtained from $\mathbf{A}_n$ with $\rho_n$ assumed known. If $\rho_n \equiv 1$ for all $n$, then
\begin{gather}
\label{eq:naive_clt1}
n (\hat{\mathbf{B}}^{(N)}_{kk} - \mathbf{B}_{kk}) \overset{\mathrm{d}}{\longrightarrow} \mathcal{N}\Bigl(0, \frac{2 \mathbf{B}_{kk} (1 - \mathbf{B}_{kk})}{\pi_k^2}\Bigr), \quad \text{for $k \in [K]$} \\
\label{eq:naive_clt2}
n(\hat{\mathbf{B}}^{(N)}_{kl} - \mathbf{B}_{kl}) \overset{\mathrm{d}}{\longrightarrow} \mathcal{N}\Bigl(0, \frac{\mathbf{B}_{kl} (1 - \mathbf{B}_{kl})}{\pi_k \pi_l}\Bigr), \quad \text{for $k \in [K], l \in [K], k \not = l$}
\end{gather}
as $n \rightarrow \infty$, and that the $K(K+1)/2$ random variables $\{n(\hat{\mathbf{B}}^{(N)}_{kl}- \mathbf{B}_{kl})\}_{k \leq l}$ are asymptotically independent. If, however, $\rho_n \rightarrow 0$ with $n \rho_n = \omega(\log{n})$, then
\begin{gather}
\label{eq:naive_clt3}
n \sqrt{\rho_n} (\hat{\mathbf{B}}^{(N)}_{kk} - \mathbf{B}_{kk}) \overset{\mathrm{d}}{\longrightarrow} \mathcal{N}\Bigl(0, \frac{2 \mathbf{B}_{kk}}{\pi_k^2}
 \Bigr), \quad \text{for $k \in [K]$} \\
\label{eq:naive_clt4}
n \sqrt{\rho_n} (\hat{\mathbf{B}}^{(N)}_{kl} - \mathbf{B}_{kl}) \overset{\mathrm{d}}{\longrightarrow} \mathcal{N}\bigl(0, \frac{\mathbf{B}_{kl}}{\pi_k \pi_l}\Bigr), \quad \text{for $k \in [K], l \in [K], k \not = l$}
\end{gather}
as $n \rightarrow \infty$, and the $K(K+1)/2$ random variables $\{n \sqrt{\rho_n}(\hat{\mathbf{B}}^{(N)}_{kl} - \mathbf{B}_{kl} \}_{k \leq l}$ are asymptotically independent.
\end{theorem}
Furthermore, $\hat{\mathbf{B}}^{(N)}$ is also purported to be optimal; Section~5 of \cite{bickel_asymptotic_normality} states that
\begin{quotation}
These results easily imply that classical optimality properties of these procedures\footnote{The procedures referred to here are the MLE $\hat{\mathbf{B}}^{(N)}$ and its variational approximation as introduced in \cite{daudin}.}, such as achievement of the information bound, hold.
\end{quotation}
We present a few examples in Section~\ref{sec:main_results} illustrating that $\hat{\mathbf{B}}^{(N)}$ is optimal only if $\mathbf{B}$ is invertible, and that in general, for singular $\mathbf{B}$, $\hat{\mathbf{B}}^{(N)}$ is dominated by the MLE $\hat{\mathbf{B}}^{(M)}$.

Another widely used technique, and computationally tractable alternative to $\hat{\mathbf{B}}^{(N)}$, for estimating $\mathbf{B}$ is based on spectral embedding $\mathbf{A}$. More specifically, given $\mathbf{A} \sim \mathrm{SBM}(\mathbf{B}, \bm{\pi})$ with known sparsity factor $\rho$, we consider the following procedure for estimating $\mathbf{B}$:
\begin{enumerate}
\item Assuming $d = \mathrm{rk}(\mathbf{B})$ is known, let $\mathbf{A} = \hat{\mathbf{U}} \hat{\bm{\Lambda}} \hat{\mathbf{U}}^{\top} + \hat{\mathbf{U}}_{\perp} \hat{\bm{\Lambda}}_{\perp} \hat{\mathbf{U}}^{\top}_{\perp}$ be the eigendecomposition of $\mathbf{A}$ where $\hat{\bm{\Lambda}}$ is the diagonal matrix containing the $d$ largest eigenvalues of $\mathbf{A}$ in modulus and $\hat{\mathbf{U}}$ is the $n \times d$ matrix whose columns are the corresponding eigenvectors of $\mathbf{A}$. 

\item Assuming $K$ is known, cluster the rows of $\hat{\mathbf{U}}$ into $K$ clusters using $K$-means, obtaining an ``estimate'' $\hat{\bm{\tau}}$ of $\bm{\tau}$.

\item For $k \in [K]$, let $\hat{\bm{s}}_{k}$ be the vector in $\mathbb{R}^{n}$ where the $i$-th entry of $\hat{\bm{s}}_k$ is $1$ if $\hat{\tau}_i = k$ and $0$ otherwise and let $\hat{n}_k = |\{i \colon \hat{\tau}_i = k\}|$ be the number of vertices assigned to block $k$. 

\item Estimate $\mathbf{B}_{kl}$ by $\hat{\mathbf{B}}^{(S)}_{kl} = \frac{1}{\hat{n}_k \hat{n}_l \rho∫} \hat{\bm{s}}_k^{\top} \hat{\mathbf{U}} \hat{\bm{\Lambda}} \hat{\mathbf{U}}^{\top} \hat{\bm{s}}_{l}$.
\end{enumerate}
The above procedure assumes $d $ and $K$ are known. When $d$ is unknown, it can be consistently estimated using the following approach: let $\hat{d}$ be the number of eigenvalues of $\mathbf{A}$ exceeding $4 \sqrt{\delta(\mathbf{A})}$ in modulus; here $\delta(\mathbf{A})$ denotes the max degree of $\mathbf{A}$. Then $\hat{d}$ is a consistent estimate of $d$ as $n \rightarrow \infty$. This follows directly from tail bounds for $\|\mathbf{A} - \mathbb{E}[\mathbf{A}]\|$ (see e.g., \cite{rinaldo_2013,oliveira2009concentration}) and Weyl's inequality. The estimation of $K$ is investigated in \cite{bickel13:_hypot,lei2014} among others, and is based on showing that if $\mathbf{A}$ is a $K$-block SBM, then there exists a consistent estimate $\widehat{\mathbb{E}[\mathbf{A}]} = (\hat{\bA}_{ij})$ of $\mathbb{E}[\mathbf{A}]$ such that the matrix $\tilde{\mathbf{A}}$ with entries $\tilde{\bA}_{ij} = (\bA_{ij} - \hat{\bA}_{ij})/\sqrt{(n-1)\hat{\bA}_{ij} (1 - \hat{\bA}_{ij})}$ has a limiting Tracy-Widom distribution, i.e., $n^{2/3}(\lambda_{1}(\tilde{\mathbf{A}}) - 2)$ converges to Tracy-Widom.

The above spectral embedding procedure (and related procedures based on eigendecomposition of other matrices such as the normalized Laplacian) is also well-studied, see e.g., \cite{rinaldo_2013,sussman12,rohe2011spectral,mcsherry,joseph_yu_2015,bickel_sarkar_2013,fishkind2013consistent,athreya2013limit,tang_priebe_16,perfect,coja-oghlan} among others; however, the available results mostly focused on showing that $\hat{\bm{\tau}}$ consistently recovers $\bm{\tau}$. Since $\hat{\bm{\tau}}$ is almost surely an exact recovery of $\bm{\tau}$ in the limit, i.e., $\hat{\tau}_i = \tau_i$ for all $i$ as $n \rightarrow \infty$ (see e.g., \cite{perfect,mcsherry}), the estimate of $\mathbf{B}_{k \ell}$ given by $\tfrac{1}{\hat{n}_k \hat{n}_{\ell} \rho} \hat{\bm{s}}_k^{\top} \mathbf{A} \hat{\bm{s}}_{\ell}$ is a consistent estimate of $\mathbf{B}$, and furthermore, coincides with $\hat{\mathbf{B}}^{(N)}$ as $n \rightarrow \infty$. 
The quantity $\hat{\mathbf{U}} \hat{\bm{\Lambda}} \hat{\mathbf{U}}^{\top} - \mathbb{E}[\mathbf{A}]$ is also widely analyzed in the context of matrix and graphon estimation using universal singular values thresholding \cite{chatterjee,gao,xu_spectral,klopp}; the focus there had been in showing the minimax rates of convergence of $\tfrac{1}{n^2}\|\hat{\mathbf{U}} \hat{\bm{\Lambda}} \hat{\mathbf{U}}^{\top} - \mathbb{E}[\mathbf{A}]\|_{F}$ to $0$ as $n \rightarrow \infty$. These rates of convergence, however, do not translate to results on the limiting distribution of $\hat{\mathbf{B}}^{(S)}_{k \ell} - \mathbf{B}_{k \ell} = \tfrac{1}{\hat{n}_k \hat{n}_{\ell} \rho} \hat{\bm{s}}_k^{\top} \hat{\mathbf{U}} \hat{\bm{\Lambda}} \hat{\mathbf{U}}^{\top} \hat{\bm{s}}_{\ell} - \tfrac{1}{n_{k} n_{\ell} \rho} \bm{s}_k^{\top}\mathbb{E}[\mathbf{A}] \bm{s}_{\ell}$.

In summary, formal comparisons between $\hat{\mathbf{B}}^{(N)}$ and $\hat{\mathbf{B}}^{(S)}$ are severely lacking. Our paper addresses this important void in the literature. The contributions of our paper are as follows. For stochastic blockmodel graphs with sparisty factors $\rho_n$ satisfying $n \rho_n = \omega(\sqrt{n})$, we establish asymptotic normality of $n \rho_n^{1/2} (\hat{\mathbf{B}}^{(S)} - \mathbf{B})$ in Theorem~\ref{THM:GEN_D} and Theorem~\ref{THM:GEN_D_SPARSE}. As a corollary of this result, we show that when $\mathbf{B}$ is of full-rank, that $n \sqrt{\rho_n}(\hat{\mathbf{B}}^{(S)} - \mathbf{B})$ has the same limiting distribution as $n \sqrt{\rho_n} (\hat{\mathbf{B}}^{(N)} - \mathbf{B})$ given in Eq.~\eqref{eq:naive_clt1} and Eq.~\eqref{eq:naive_clt2} and that both estimators are asymptotically efficient; the two estimators $\hat{\mathbf{B}}^{(M)}$ and $\hat{\mathbf{B}}^{(N)}$ are identical in this setting. When $\mathbf{B}$ is singular, we show that $n \sqrt{\rho_n}(\hat{\mathbf{B}}^{(S)} - \mathbf{B})$ can have smaller variances than $n \sqrt{\rho_n} (\hat{\mathbf{B}}^{(N)} - \mathbf{B})$, and thus a bias-corrected $\hat{\mathbf{B}}^{(S)}$ can have smaller mean square error than $\hat{\mathbf{B}}^{(N)}$, and furthermore, that the resulting bias-corrected $\hat{\mathbf{B}}^{(S)}$ can be almost as efficient as the asymptotically first-order efficient estimator $\hat{\mathbf{B}}^{(M)}$. Finally, we also provide some justification of the potential necessity of the condition that the average degree satisfies $n \rho_n = \omega(\sqrt{n})$; in essence, as $\rho_n \rightarrow 0$, the bias incurred by the low-rank representation $
\hat{\bU} \hat{\bm{\Lambda}} \hat{\bU}^{\top}$ of $\mathbf{A}$ overwhelms the reduction in variance resulting from the low-rank representation.

\section{Central limit theorem for $\hat{\mathbf{B}}^{(S)}_{k\ell}$}
\label{sec:main_results}
Let $(\mathbf{A}, \bm{\tau}) \sim \mathrm{SBM}(\mathbf{B}, \bm{\pi}, \rho)$ be a stochastic blockmodel graph 
on $n$ vertices with sparsity factor $\rho$ and $\mathrm{rk}(\mathbf{B}) = d$.
We first consider the setting wherein both $\bm{\tau}$ and $d$ are assumed known. The setting wherein $\bm{\tau}$ is unobserved and needs to be recovered will be addressed subsequently in Corollary~\ref{cor:in_practice}, while the setting when $d$ is unknown was previously addressed following the introduction of the estimator $\hat{\mathbf{B}}^{(S)}$ in Section~\ref{sec:intro}. When $\bm{\tau}$ is known, the spectral embedding estimate of $\mathbf{B}_{kl}$ (with $\rho$ assumed known) is $\hat{\mathbf{B}}_{kl}^{(S)} = \tfrac{1}{\rho n_k n_l}
\bm{s}_k^{\top} \hat{\mathbf{U}} \hat{\bm{\Lambda}} \hat{\mathbf{U}}^{\top} \bm{s}_l$ where $\bm{s}_k$ is the vector whose elements $\{s_{ki}\}$ are such that 
$s_{ki} = 1$ if vertex $i$ is assigned to block $k$ and $s_{ki} = 0$ otherwise; here $n_k$ denote the number of vertices $v_i$ assigned to block $k$.

We then have the following non-degenerate limiting distribution of $\hat{\mathbf{B}}^{(S)} - \mathbf{B}$. We shall present two variants of this limiting distribution. The first variant, Theorem~\ref{THM:GEN_D}, applies to the setting where the average degree grows linearly with $n$, i.e., the 
sparsity factor $\rho_n \rightarrow c > 0$; without loss of generality, we can assume $\rho_n \equiv c = 1$. The second variant applies to the setting where the average degree grows sub-linearly in $n$, i.e., $\rho_n \rightarrow 0$ with $n \rho_n = \omega(\sqrt{n})$. For ease of exposition, these variants (and their proofs) will be presented using the following parametrization of stochastic blockmodel graphs as a sub-class of the more general random dot product graphs model \cite{young2007random,grdpg1}.
\begin{definition}[Generalized random dot product graph]
\label{def:grdpg}
  Let $d$ be a positive integer and $p \geq 1$ and $q \geq 0$ be 
  such that $p + q = d$. Let $\mathbf{I}_{p,q}$ denote the
  diagonal matrix whose diagonal elements contains $p$ entries
  equaling $1$ and $q$ entries equaling $-1$. 
  Let $\mathcal{X}$ be a subset of 
  $\mathbb{R}^{d}$ such $x^{\top} \mathbf{I}_{p,q} y \in[0,1]$ for all $x,y\in
  \mathcal{X}$. Let $F$ be a distribution taking values in $\mathcal{X}$. 
  We say $(\mathbf{X},\mathbf{A}) \sim
  \mathrm{GRDPG}_{p,q}(F)$ with sparsity factor $\rho \in (0,1]$ if the following hold. First let $X_1, X_2,
  \dots, X_n \overset{\mathrm{i.i.d}}{\sim} F$ and set
  $\mathbf{X}=[X_1 \mid \cdots \mid X_n]^\top\in \mathbb{R}^{n\times
    d}$. Then $\mathbf{A}\in\{0,1\}^{n\times n}$ is a symmetric matrix
such that, conditioned on $\mathbf{X}$, for all $i \geq j$ the
$A_{ij}$ are independent and
\begin{equation}
 A_{ij} \sim \mathrm{Bernoulli}(\rho X_i^\top \mathbf{I}_{p,q} X_j).
\end{equation}
We therefore have
\begin{equation}
\Pr[\mathbf{A} \mid \mathbf{X}]=\prod_{i \leq j} (\rho X^{\top}_i \mathbf{I}_{p,q}
X_j)^{A_{ij}}(1- \rho X^{\top}_i \mathbf{I}_{p,q} X_j)^{(1-A_{ij})}.
\end{equation}
\end{definition}
It is straightforward to show that any stochastic blockmodel graph
$(\mathbf{A}, \bm{\tau}) \sim \mathrm{SBM}(\bm{\pi}, \mathbf{B})$
can also be represented as a (generalized) random dot product
graph $(\mathbf{X}, \mathbf{A}) \sim \mathrm{GRDPG}_{p,q}(F)$ where $F$ is a
mixture of point masses. Indeed, suppose $\mathbf{B}$ is a $K \times
K$ matrix and let $\mathbf{B} = \mathbf{U}
\bm{\Sigma} \mathbf{U}^{\top}$ be the eigendecomposition of
$\mathbf{B}$. Then, denoting by $\nu_1, \nu_2, \dots, \nu_K$ the rows
of $\mathbf{U} |\bm{\Sigma}|^{1/2}$, we can define $F = \sum_{k=1}^{K} \pi_k
\delta_{\nu_k}$ where $\delta$ is the Dirac delta function; $p$ and $q$ are given by the number of positive and negative
eigenvalues of $\mathbf{B}$, respectively. 

\begin{theorem}
\label{THM:GEN_D}
Let $\mathbf{A} \sim \mathrm{SBM}(\bm{\pi}, \mathbf{B}, \rho_n)$ be a $K$-block stochastic blockmodel graph on $n$ vertices with sparsity factor $\rho_n = 1$. Let $\nu_1, \dots, \nu_K$ be point masses in $\mathbb{R}^{d}$ such that $\mathbf{B}_{k \ell} = \nu_k^{\top} \mathbf{I}_{p,q} \nu_{\ell}$ and let $\Delta = \sum_{k} \pi_k \nu_k \nu_k^{\top}$. 
For $k \in [K]$ and $\ell \in [K]$, let $\theta_{k \ell}$ be given by
 \begin{equation}
 \label{eq:mu_kl}
 \begin{split}
\theta_{k\ell} &= \sum_{r=1}^{K} \pi_r \bigl(\mathbf{B}_{kr} (1 - \mathbf{B}_{kr}) + \mathbf{B}_{\ell r}(1 - \mathbf{B}_{\ell r})\bigr)\nu_k^{\top} \Delta^{-1} \mathbf{I}_{p,q} \Delta^{-1} \nu_{\ell} \\
& - \sum_{r=1}^{K} \sum_{s=1}^{K} \pi_r \pi_s \mathbf{B}_{sr} (1 - \mathbf{B}_{sr}) \nu_s^{\top} \Delta^{-1} \mathbf{I}_{p,q} \Delta^{-1} 
(\nu_{\ell} \nu_k^{\top} + \nu_k \nu_{\ell}^{\top}) \Delta^{-1} \nu_s.
 \end{split}
 \end{equation}
 Now let $\zeta_{k \ell} = \nu_{k}^{\top} \Delta^{-1} \nu_{\ell}$. Define $\sigma_{kk}^2$ for $k \in [K]$ to be
 \begin{equation}
 \begin{split}
 \label{eq:sigma_kk}
 \sigma_{k k}^2 &= 
 4 \mathbf{B}_{kk}(1 - \mathbf{B}_{kk}) \zeta_{kk}^2 + 4 \sum_{r} \pi_r \mathbf{B}_{kr} (1 - \mathbf{B}_{kr}) \zeta_{kr}^2 \bigl(\tfrac{1}{\pi_k} - 2 \zeta_{kk}\bigr) \\ &
+ 2 \sum_{r} \sum_{s} \pi_r \pi_s \mathbf{B}_{rs} (1 - \mathbf{B}_{rs}) \zeta_{kr}^2 \zeta_{ks}^2
 \end{split}
 \end{equation}
 and define $\sigma_{k \ell}^{2}$ for $k \in [K], \ell \in [K], k \not = \ell$ to be 
 \begin{equation}
 \begin{split}
 \label{eq:sigma_kl2}
 \sigma_{k\ell}^2 &= 
 \bigl(\mathbf{B}_{kk} (1 - \mathbf{B}_{kk}) + \mathbf{B}_{\ell \ell}(1 - \mathbf{B}_{\ell \ell}) \bigr)\zeta_{k \ell}^2 + 
 2 \mathbf{B}_{k \ell} (1 - \mathbf{B}_{k \ell}) \zeta_{kk} \zeta_{\ell \ell} 
\\ &+ 
\sum_{r}  \pi_{r} \mathbf{B}_{k r}(1 - \mathbf{B}_{kr}) \zeta_{\ell r}^2 \bigl(\tfrac{1}{\pi_{k}} - 2 \zeta_{kk} \bigr) 
\\ &+ \sum_{r} \pi_r \mathbf{B}_{\ell r} (1 - \mathbf{B}_{\ell r}) \zeta_{kr}^{2} \bigl(\tfrac{1}{\pi_{\ell}} - 2 \zeta_{\ell \ell} \bigr) \\
& - 2 \sum_{r}  \pi_r \bigl(\mathbf{B}_{k r} (1 - \mathbf{B}_{k r}) + \mathbf{B}_{\ell r} (1 - \mathbf{B}_{\ell r})\bigr) \zeta_{kr} \zeta_{r \ell} \zeta_{k \ell} \\
& + \frac{1}{2} \sum_{r} \sum_{s} \pi_r \pi_s
\mathbf{B}_{rs} (1 - \mathbf{B}_{rs}) (\zeta_{k r} \zeta_{\ell s} + \zeta_{\ell r} \zeta_{k s})^2. 
 \end{split}
 \end{equation}
Then for any $k \in [K]$ and $\ell \in [K]$, 
\begin{equation}
\label{eq:sbm_normal2}
  n \bigl(\hat{\mathbf{B}}^{(S)}_{k\ell} - \mathbf{B}_{k \ell} - \frac{\theta_{kl}}{n}\bigr) \overset{\mathrm{d}}{\longrightarrow} N(0, \sigma_{k \ell}^2) 
\end{equation}
as $n \rightarrow \infty$. 
\end{theorem}

\begin{theorem}
\label{THM:GEN_D_SPARSE}
Let $\mathbf{A} \sim \mathrm{SBM}(\bm{\pi}, \mathbf{B}, \rho_n)$ be a $K$-block stochastic blockmodel graph on $n$ vertices with sparsity factor $\rho_n$.
Let $\nu_1, \dots, \nu_K$ be point masses in $\mathbb{R}^{d}$ such that $\mathbf{B}_{k \ell} = \nu_k^{\top} \mathbf{I}_{p,q} \nu_{\ell}$ and let $\Delta = \sum_{k} \pi_k \nu_k \nu_k^{\top}$. 
For $k \in [K]$ and $\ell \in [K]$, let $\tilde{\theta}_{k \ell}$ be given by
\begin{equation}
\begin{split}
\label{eq:tilde_theta_kl}
\tilde{\theta}_{k\ell} & = \sum_{r=1}^{K} \pi_r \bigl(\mathbf{B}_{kr} + \mathbf{B}_{\ell r} \bigr)\nu_k^{\top} \Delta^{-1} \mathbf{I}_{p,q} \Delta^{-1} \nu_{\ell} 
 \\ &- \sum_{r=1}^{K} \sum_{s=1}^{K} \pi_r \pi_s \mathbf{B}_{sr} \nu_s^{\top} \Delta^{-1} \mathbf{I}_{p,q} \Delta^{-1} 
(\nu_{\ell} \nu_k^{\top} + \nu_k \nu_{\ell}^{\top}) \Delta^{-1} \nu_s,
\end{split}
\end{equation}
let $\tilde{\sigma}_{kk}^{2}$ for $k \in [K]$ be
\begin{equation}
\begin{split}
\label{eq:tilde_sigma_kk}
\tilde{\sigma}_{kk}^{2} &=
4 \mathbf{B}_{kk} \zeta_{kk}^2 + 4 \sum_{r} \pi_r \mathbf{B}_{kr} \zeta_{kr}^2 \bigl(\tfrac{1}{\pi_k} - 2 \zeta_{kk}\bigr)^{2} 
\\ &+ 2 \sum_{r} \sum_{s} \pi_r \pi_s \mathbf{B}_{rs} \zeta_{kr}^2 \zeta_{ks}^2,
\end{split}
\end{equation}
and let $\tilde{\sigma}_{k\ell}^{2}$ for $k \in [K]$, $\ell \in [K]$, $k \not = \ell$ be
\begin{equation}
 \begin{split}
 \label{eq:tilde_sigma_kl}
 \tilde{\sigma}_{k\ell}^2 &= 
 \bigl(\mathbf{B}_{kk}  + \mathbf{B}_{\ell \ell} \bigr)\zeta_{k \ell}^2 +
2 \mathbf{B}_{k \ell} \zeta_{kk} \zeta_{\ell \ell} 
 - 2 \sum_{r}  \pi_r \bigl(\mathbf{B}_{k r}  + \mathbf{B}_{\ell r} \bigr) \zeta_{kr} \zeta_{r \ell} \zeta_{k \ell} 
\\ &  + 
\sum_{r}  \pi_{r} \mathbf{B}_{k r} \zeta_{\ell r}^2 \bigl(\tfrac{1}{\pi_{k}} - 2 \zeta_{kk} \bigr)^{2} 
+ \sum_{r} \pi_r \mathbf{B}_{\ell r} \zeta_{kr}^{2} \bigl(\tfrac{1}{\pi_{\ell}} - 2 \zeta_{\ell \ell} \bigr)^{2} 
\\ &+ \frac{1}{2} \sum_{r} \sum_{s} \pi_r \pi_s
\mathbf{B}_{rs} (\zeta_{k r} \zeta_{\ell s} + \zeta_{\ell r} \zeta_{k s})^2. 
 \end{split}
 \end{equation}
If $\rho_n \rightarrow 0$ and $n \rho_n = \omega(\sqrt{n})$, then for any $k \in [K]$ and $\ell \in [K]$, 
\begin{equation}
\label{eq:sbm_normal3}
  n \sqrt{\rho_n} \bigl(\hat{\mathbf{B}}^{(S)}_{k\ell} - \mathbf{B}_{k \ell} - \frac{\tilde{\theta}_{kl}}{n \rho_n}\bigr) \overset{\mathrm{d}}{\longrightarrow} N(0, \tilde{\sigma}_{k \ell}^2) 
\end{equation}
as $n \rightarrow \infty$.
\end{theorem}
The proofs of Theorem~\ref{THM:GEN_D} and Theorem~\ref{THM:GEN_D_SPARSE} are given in the appendix. As a corollary of Theorem~\ref{THM:GEN_D} and Theorem~\ref{THM:GEN_D_SPARSE}, we have the following result for the asymptotic efficiency of $\hat{\mathbf{B}}^{(S)}$ whenever $\mathbf{B}$ is invertible (see also Theorem~\ref{thm:bickel_choi}).
\begin{corollary}
\label{cor:full-rank}
Let $\mathbf{A} \sim \mathrm{SBM}(\bm{\pi}, \mathbf{B}, \rho_n)$ be a $K$-block stochastic blockmodel graph 
 on $n$ vertices with sparsity factor $\rho_n$. Suppose $\mathbf{B}$ is invertible. Then for all $k \in [K], \ell \in [K]$, $\sigma_{k \ell}$ and $\tilde{\sigma}_{k \ell}$ as defined in Theorem~\ref{THM:GEN_D} and Theorem~\ref{THM:GEN_D_SPARSE} satisfy
 \begin{gather}
 \theta_{k \ell} = 0; \quad \tilde{\theta}_{k \ell} = 0 \\
 \sigma_{kk}^2 = \frac{2 \mathbf{B}_{kk} (1 - \mathbf{B}_{kk})}{\pi_k^2}; \quad
 \sigma_{k\ell}^2 = \frac{\mathbf{B}_{k \ell} (1 - \mathbf{B}_{k\ell})}{\pi_k \pi_{\ell}} \,\, \text{if $k \not = \ell$} \\
 \tilde{\sigma}_{kk} = \frac{2 \mathbf{B}_{kk}}{\pi_k^2}; \quad \tilde{\sigma}_{k \ell} = \frac{2 \mathbf{B}_{k \ell}}{\pi_k \pi_{\ell}} \,\, \text{if $k \not = \ell$}.
 \end{gather}
 Therefore, for all $k \in [K], \ell \in [K]$, if $\rho_n \equiv 1$, then
 \begin{equation}
  n (\hat{\mathbf{B}}^{(S)}_{k\ell} - \mathbf{B}_{k\ell}) \overset{\mathrm{d}}{\longrightarrow} N(0, \sigma_{k\ell}^2).
 \end{equation}
 as $n \rightarrow \infty$. If $\rho_n \rightarrow 0$ and $n \rho_n = \omega(\sqrt{n})$, then
  \begin{equation}
  n \rho_n^{1/2}(\hat{\mathbf{B}}^{(S)}_{k\ell} - \mathbf{B}_{k\ell}) \overset{\mathrm{d}}{\longrightarrow} N(0, \tilde{\sigma}_{k\ell}^2).
 \end{equation}
as $n \rightarrow \infty$. $\hat{\mathbf{B}}^{(S)}_{k \ell}$ is therefore {\em asymptotically efficient} for all $k, \ell$.
\end{corollary}
\begin{proof}[Proof of Corollary~\ref{cor:full-rank}]
The proof follows trivially from the observation that $\zeta_{r s} = \tfrac{1}{\pi_r}$ for $r = s$ and $\zeta_{rs} = 0$ otherwise. 
Indeed, $\Delta = \sum_{k} \pi_k \nu_k \nu_k^{\top} = \bm{\nu}^{\top} \mathbf{D} \bm{\nu}$ 
where $\bm{\nu}$ is a $K \times K$ matrix with $\bm{\nu} \mathbf{I}_{p,q} \bm{\nu}^{\top} = \mathbf{B}$ and $\mathbf{D} = \mathrm{diag}(\bm{\pi})$. Hence
\begin{equation} 
\label{key:reduction}
\zeta_{rs} = \nu_r^{\top} \Delta^{-1} \nu_{s} = \nu_r^{\top} \bm{\nu}^{-1} \mathbf{D}^{-1} (\bm{\nu}^{-1})^{\top} \nu_s = \frac{1}{\pi_r} \mathbbm{1}\{r = s\}.
\end{equation}
As an example, the expression for $\theta_{k \ell}$ in Eq.~\eqref{eq:mu_kl} reduces to
\begin{equation*}
\begin{split} \theta_{k\ell} &= \sum_{r=1}^{K} \pi_r \bigl(\mathbf{B}_{kr} (1 - \mathbf{B}_{kr}) + \mathbf{B}_{\ell r}(1 - \mathbf{B}_{\ell r})\bigr)\nu_k^{\top} \Delta^{-1} \mathbf{I}_{p,q} \Delta^{-1} \nu_{\ell}  \\ &- \sum_{r=1}^{K} \sum_{s=1}^{K} \pi_r \pi_s \mathbf{B}_{sr} (1 - \mathbf{B}_{sr}) \nu_s^{\top} \Delta^{-1} \mathbf{I}_{p,q} \Delta^{-1}
 \Bigl(\tfrac{\mathbbm{1}\{s = k\}}{\pi_k} \nu_{\ell} + \tfrac{\mathbbm{1}\{s = \ell\}}{\pi_{\ell}} \nu_{k} \Bigr) 
\\ &= 0
\end{split}
\end{equation*}
for all $k \in [K], \ell \in [K]$. The expression for $\sigma_{kk}^{2},\tilde{\sigma}_{kk}^2$, $\sigma_{k \ell}^2$ and $\tilde{\sigma}_{k\ell}^2$ also follows directly from Eq.~\eqref{key:reduction}. 
\end{proof}

\begin{remark}
As a special case of Theorem~\ref{THM:GEN_D}, we consider the two blocks stochastic blockmodel with block probability matrix $\mathbf{B} = \Bigl[\begin{smallmatrix} p^2 & pq \\ pq & q^2 \end{smallmatrix} \Bigr]$ and block assignment probabilities $\bm{\pi} = (\pi_p, \pi_q), \pi_p + \pi_q = 1.$ Then 
$\Delta = \pi_p p^2 + \pi_q q^2$ and Eq.~\eqref{eq:mu_kl} reduces to 
\begin{gather*}
\theta_{11} = \tfrac{2 \pi_q p^2 q^2}{\Delta^3} \bigl(\pi_p p^2 (1 - p^2) + (\pi_q - \pi_p) p q (1 - pq) -  \pi_q q^2 (1 - q^2) \bigr), \\
\theta_{12} = \tfrac{pq}{\Delta^3} \bigl(\pi_p p^2 ( 1- p^2) (\pi_q q^2 - \pi_p p^2) + (\pi_p - \pi_q) pq (1 - pq) (\pi_p q^2 - \pi_q p^2) \\ + \pi_q q^2 ( 1- q^2) (\pi_p p^2 - \pi_q q^2) \bigr),	 \\
\theta_{22} = \tfrac{2 \pi_p p^2 q^2}{\Delta^3} \bigl(\pi_q q^2 (1 - q^2) + (\pi_p - \pi_q) p q (1 - pq) - \pi_p p^2 (1 - p^2) \bigr). 
\end{gather*}
Meanwhile, we also have
\begin{gather*}
\sigma^{2}_{11} = \tfrac{8 p^6 (1 - p^2)}{\Delta^2} \bigl(1 - \tfrac{\pi_p p^2}{2 \Delta}\bigr)^2 + \tfrac{4 \pi_q p^3 q^3 (1 - pq)}{\pi_p \Delta^2} \bigl(1 - \tfrac{\pi_p p^2}{\Delta}\bigr)^2 + \tfrac{2 \pi_q^2 p^4 q^6 (1 - q^2)}{\Delta^4} \\
\sigma^{2}_{12} = \tfrac{2 \pi_q^2 p^4 q^6 (1 - p^2)}{\Delta^4} + \tfrac{\pi_p \pi_q pq (1 - pq)}{\Delta^4} \bigl(\tfrac{\pi_q q^4}{\pi_p} + \tfrac{\pi_p p^4}{\pi_q}\bigr)^2 + \tfrac{2 \pi_p^2 p^6 q^4 (1 - q^2)}{\Delta^4} \\
\sigma^2_{22} = \tfrac{8 q^6 (1 - q^2)}{\Delta^2} \bigl(1 - \tfrac{\pi_q q^2}{2 \Delta}\bigr)^2 + \tfrac{4 \pi_p p^3 q^3 (1 - pq)}{\pi_q \Delta^2} \bigl(1 - \tfrac{\pi_q q^2}{\Delta}\bigr)^2 + \tfrac{2 \pi_p^2 q^4 p^6 (1 - p^2)}{\Delta^4}.
\end{gather*}
The naive (MLE) estimator $\hat{\mathbf{B}}^{(N)}$ has asymptotic variances 
\begin{gather*}
\mathrm{Var}[\hat{\mathbf{B}}^{(N)}_{11}] = \tfrac{2 p^2 (1 - p^2)}{\pi_p^2}; \quad \mathrm{Var}[\hat{\mathbf{B}}^{(N)}_{12}] = \tfrac{pq (1 - pq)}{\pi_p \pi_q}; \quad \mathrm{Var}[\hat{\mathbf{B}}^{(N)}_{22}] = \tfrac{2 q^2 (1 - q^2)}{\pi_q^2}.
\end{gather*}
We now evaluate the asymptotic variances for the true (rank-constrained) MLE. Suppose for simplicity that $n \pi_p$ vertices are assigned to block $1$ and $n \pi_q$ vertices are assigned to block $2$. Let $n_{11} = \tbinom{(n \pi_p + 1)}{2}$, $n_{12} = n^2 \pi_p \pi_q$ and $n_{22} = \tbinom{(n \pi_p + 1)}{2}$. 
Let $\mathbf{A}$ be given and assume for the moment that $\bm{\tau}$ is observed. Then the log-likelihood for $\mathbf{A}$ is equivalent to the log-likelihood for observing $m_{11} \sim \mathrm{Bin}(n_{11}, p^2)$,  $m_{12} \sim \mathrm{Bin}(n_{12}, pq)$ and $m_{22} \sim \mathrm{Bin}(n_{22}, q^2)$ with $m_{11}$, $m_{12}$, and $m_{22}$ mutually independent. More specifically, ignoring terms of the form $\tbinom{n_{ij}}{m_{ij}}$, we have
\begin{equation*}
\begin{split}
 \ell(\mathbf{A} \mid p,q) &= m_{11} \log p^2 + (n_{11} - m_{11}) \log (1 - p^2) + m_{12} \log pq
 \\ &+ (n_{12} - m_{12}) \log (1 - pq) + m_{22} \log q^2 + (n_{22} - m_{22}) \log (1 - q^2) 
 \end{split}
 \end{equation*}
 We therefore have
 \begin{equation*}
 \begin{split}
 \mathrm{Var}\Bigl(\tfrac{\partial \ell}{\partial p}\Bigr) &= \mathrm{Var}\Bigl(\tfrac{2m_{11}p}{p^2} - \tfrac{2(n_{11} - m_{11})p}{1 - p^2} + 
\tfrac{m_{12}q}{pq} - \tfrac{(n_{12} - m_{12})q}{1 - pq}\Bigr)
= \tfrac{4n_{11}}{1 - p^2} + \tfrac{n_{12} pq}{p^2(1 - pq)}.
 \end{split}
 \end{equation*}
 Similarly,
 \begin{gather*}
 \mathrm{Var}\Bigl(\tfrac{\partial \ell}{\partial q}\Bigr) = \tfrac{n_{12}pq}{q^2(1-pq)} + \tfrac{4n_{22}}{1 - q^2}; \quad
 \mathrm{Cov}\Bigl(\tfrac{\partial \ell}{\partial p}, \tfrac{\partial \ell}{\partial q} \Bigr) = \tfrac{n_{12}}{1 - pq}.
 \end{gather*}
 Next we note that $n_{11}/n^2 \rightarrow \pi_p^2/2$, 
$n_{12}/n^2 \rightarrow \pi_p \pi_q$ and 
 $n_{22}/n^2 \rightarrow \pi_q^2/2$. We therefore have
 $$\frac{1}{n^2} \mathrm{Var}\Bigl[\Bigl(\tfrac{\partial \ell}{\partial p}, \tfrac{\partial \ell}{\partial p}\Bigr) \Bigr] \overset{\mathrm{a.s.}}{\longrightarrow} \mathcal{I} := \begin{bmatrix} \frac{2 \pi_p^2}{1 - p^2} + \frac{\pi_p \pi_q q}{p(1 - pq)} & \frac{\pi_p \pi_q}{1 - pq} \\
 \frac{\pi_p \pi_q}{1 - pq} & \frac{2 \pi_q^2}{1 - q^2} + \frac{\pi_p \pi_q p}{q(1 - pq)} \end{bmatrix} $$
 as $n \rightarrow \infty$. Let $(\hat{p}, \hat{q})$ be the MLE of $(p,q)$ (we emphasize that there is no close form formula for the MLE $\hat{p}$ and $\hat{q}$ in this setting) and let $\mathcal{J}$ be the Jacobian of $\mathbf{B}$ with respect to $p$ and $q$, i.e.,
$ \mathcal{J}^{\top} = \Bigl[\begin{smallmatrix} 2p & q & 0 \\ 0 & p & 2q \end{smallmatrix}\Bigr].$
 The classical theory of maximum likelihood estimation (see e.g., Theorem~5.1 of \cite{point_estimation}) implies that
 $$ n\bigl( \mathrm{vech}(\hat{\mathbf{B}}^{(M)} - \mathbf{B}) \bigr) \overset{\mathrm{d}}{\longrightarrow} \mathrm{MVN}\bigl(\bm{0}, \mathcal{J} \mathcal{I}^{-1} \mathcal{J}^{\top} \bigr)$$
 as $n \rightarrow \infty$, i.e., $\hat{\mathbf{B}}^{(M)}$ is asymptotically (first-order) efficient.

\begin{figure}[tp!]
\center
\subfloat[$\mathrm{MSE}(\hat{\bB}^{(S)})/\mathrm{MSE}(\hat{\bB}^{(N)})$]{
\includegraphics[width=0.48\textwidth]{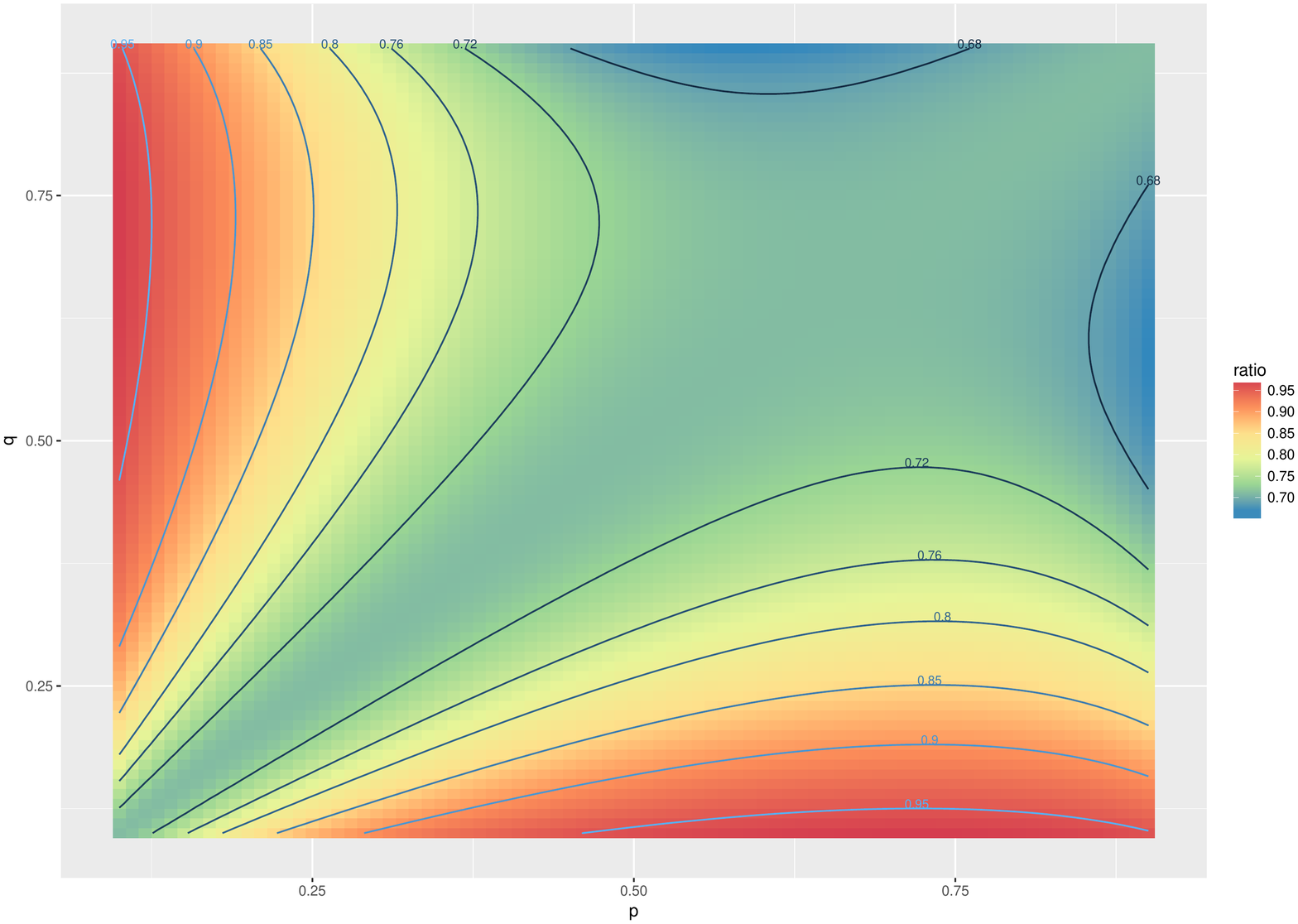}} 
\subfloat[$\mathrm{MSE}(\hat{\bB}^{(S)})/\mathrm{MSE}(\hat{\bB}^{(M)})$]{
\includegraphics[width=0.48\textwidth]{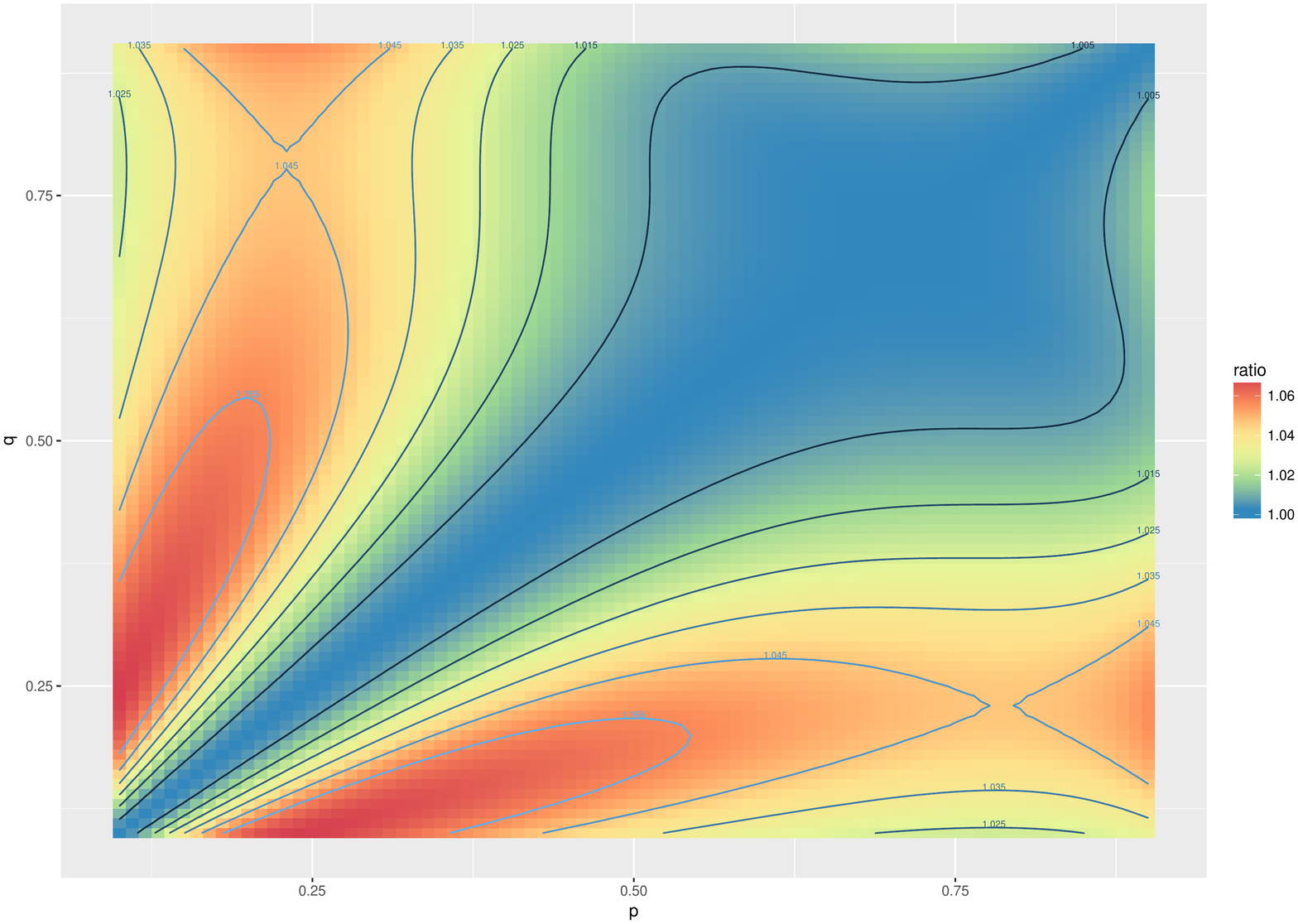}} 
\caption{The ratios $\mathrm{MSE}(\hat{\bB}^{(S)})/\mathrm{MSE}(\hat{\bB}^{(N)})$
and $\mathrm{MSE}(\hat{\bB}^{(S)})/\mathrm{MSE}(\hat{\bB}^{(M)})$ for values of 
$p \in [0.1, 0.9]$ and $q \in [0.1, 0.9]$ in a $2$-blocks rank $1$ SBM.
The labeled lines in each plot are the contour lines for the $\mathrm{MSE}$ ratios. 
The $\mathrm{MSE}(\hat{\bB}^{(S)})$ is computed using the bias-adjusted estimates $\{\hat{\mathbf{B}}^{(S)}_{ij} - \hat{\theta}_{ij}\}_{i \leq j}$ (see Corollary~\ref{cor:in_practice}).
}
\label{fig:ratio-plot}
\end{figure}
We now compare the naive estimator $\hat{\mathbf{B}}^{(N)}$, the ASE estimator $\hat{\mathbf{B}}^{(S)}$ and the true MLE $\hat{\mathbf{B}}^{(M)}$ for the case when $\pi_p  = 0.5$ and $\rho_n \equiv 1$.
Plots of the (ratio) of the MSE (equivalently the sum of variances $\sigma_{11}^2 + \sigma_{12}^2 + \sigma_{22}^2$) 
for $\hat{\mathbf{B}}^{(S)}$ (adjusted for the bias terms $\theta_{k \ell}$; see Corollary~\ref{cor:in_practice}) against the MSE (equivalently the sum of variances) of $\hat{\mathbf{B}}^{(N)}$ and the MSE of $\hat{\mathbf{B}}^{(M)}$ for
$p \in [0.1,0.9]$ and $q \in [0.1,0.9]$ are given in Figure~\ref{fig:ratio-plot}. For this example,  $\hat{\mathbf{B}}^{(S)}$ have smaller mean squared error than $\hat{\mathbf{B}}^{(N)}$ over the whole range of $p$ and $q$. In addition $\hat{\mathbf{B}}^{(S)}$ has mean squared error almost as small as that of $\hat{\mathbf{B}}^{(M)}$ for a large range of $p$ and $q$.
\end{remark}

\begin{remark}
We next compare the three estimators $\hat{\mathbf{B}}^{(S)}$, $\hat{\mathbf{B}}^{(M)}$ and $\hat{\mathbf{B}}^{(N)}$ in the setting of stochastic blockmodels with $3 \times 3$ block probability matrix $\mathbf{B}$ where $\mathbf{B}$ is positive semidefinite and $\mathrm{rk}(\mathbf{B}) = 2$. The minimal parametrization of $\mathbf{B}$ requires $5$ parameters $(r_1,r_2,r_3,\theta,\gamma)$, namely
\begin{gather*}
\mathbf{B}_{11} = r_1^2; \quad \mathbf{B}_{22} = r_2^2; \quad \mathbf{B}_{33} = r_3^2; \\ \mathbf{B}_{12} = r_1 r_2 \cos \theta; \quad \mathbf{B}_{13} = r_1 r_3 \cos \gamma ; \quad \mathbf{B}_{23} = r_2 r_3 \cos(\theta - \gamma).
\end{gather*}
Now let $\bm{\pi} = (\pi_1, \pi_2, \pi_3)$ be the block assignment probability vector. Let $\mathbf{A} \sim \mathrm{SBM}(\mathbf{B}, \bm{\pi})$ be a graph on $n$ vertices and suppose for simplicity that the number of vertices in block $i$ is $n_i = n \pi_i$. Let $n_{ii} = n_i^{2}/2$ for $i=1,2,3$ 
and $n_{ij} = n_i n_j$ if $i \not = j$. Let $m_{ij}$ for $i \leq j$ be independent random variables with $m_{ij} \sim \mathrm{Bin}(n_{ij}, \mathbf{B}_{ij})$.
Then, assuming $\bm{\tau}$ is known, the log-likelihood for $\mathbf{A}$ is  
\begin{equation*}
\begin{split}
\ell(\mathbf{A}) &= m_{11} \log(r_1^2) + (n_{11} - m_{11}) \log(1 - r_1^2) + m_{22} \log (r_2^2) + (n_{22} - m_{22}) \log (1 - r_2^2) \\ &+ m_{33} \log (r_3^2) + (n_{33} - m_{33}) \log (1 - r_3^2) + m_{12} \log (r_1 r_2 \cos \theta) \\ &+ (n_{12} - m_{12}) \log (1 - r_1 r_2 \cos \theta) + m_{13} \log (r_1 r_3 \cos \gamma) \\ & + (n_{13} - m_{13}) \log (1 - r_1 r_3 \cos \gamma) + m_{23} \log (r_2 r_3 \cos (\theta - \gamma)) \\ &+ (n_{23} - m_{23}) \cos(1 - r_2 r_3 \cos(\theta - \gamma))
\end{split}
\end{equation*}
The Fisher information matrix $\mathcal{I}$ for $(r_1, r_2, r_3, \theta, \gamma)$ in this setting is straightforward, albeit tedious, to derive. 

\begin{figure}[tp!]
\center
\subfloat[$\mathrm{MSE}(\hat{\bB}^{(S)})/\mathrm{MSE}(\hat{\bB}^{(N)})$]{
\includegraphics[width=0.48\textwidth]{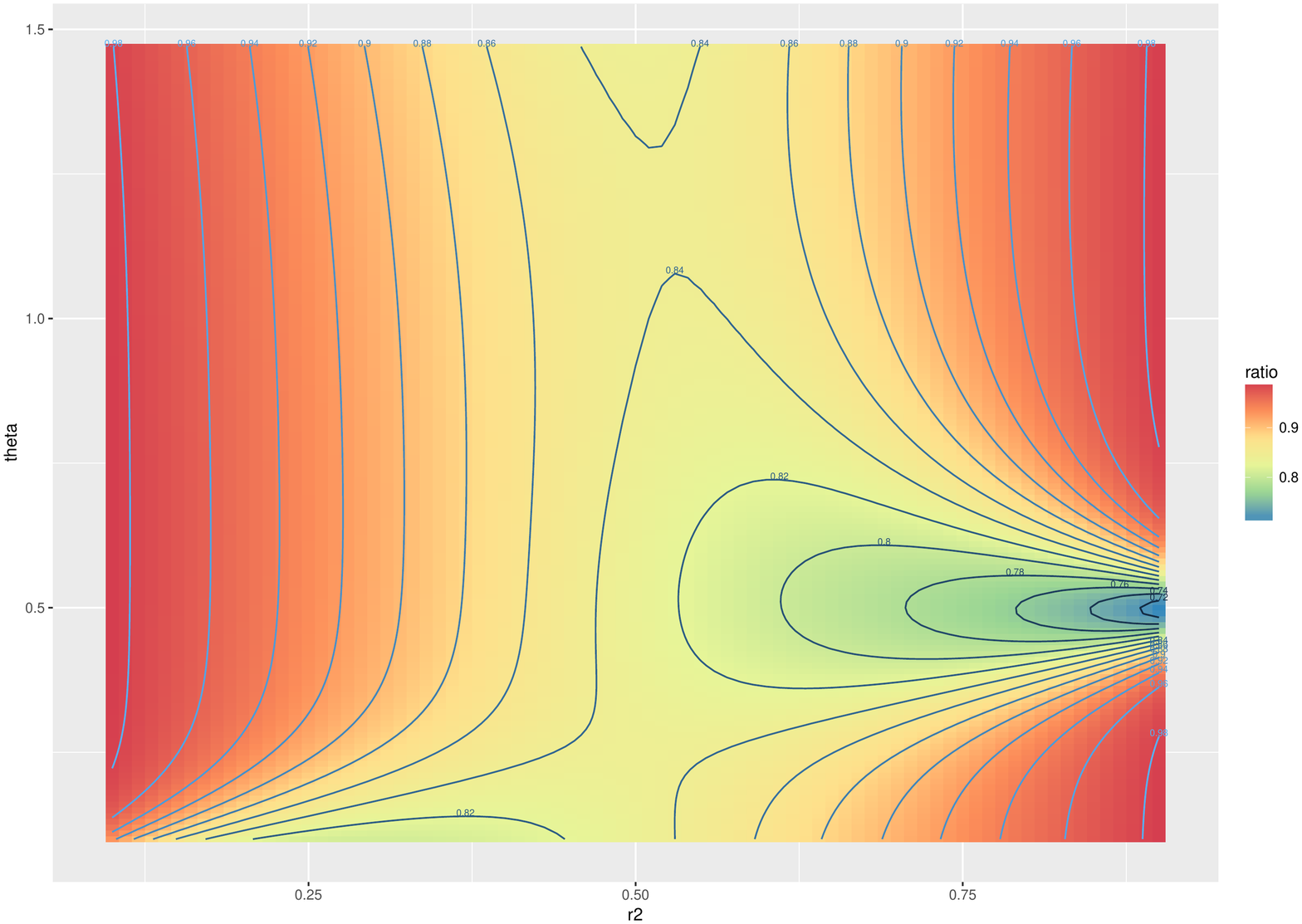}} 
\subfloat[$\mathrm{MSE}(\hat{\bB}^{(S)})/\mathrm{MSE}(\hat{\bB}^{(M)})$]{
\includegraphics[width=0.48\textwidth]{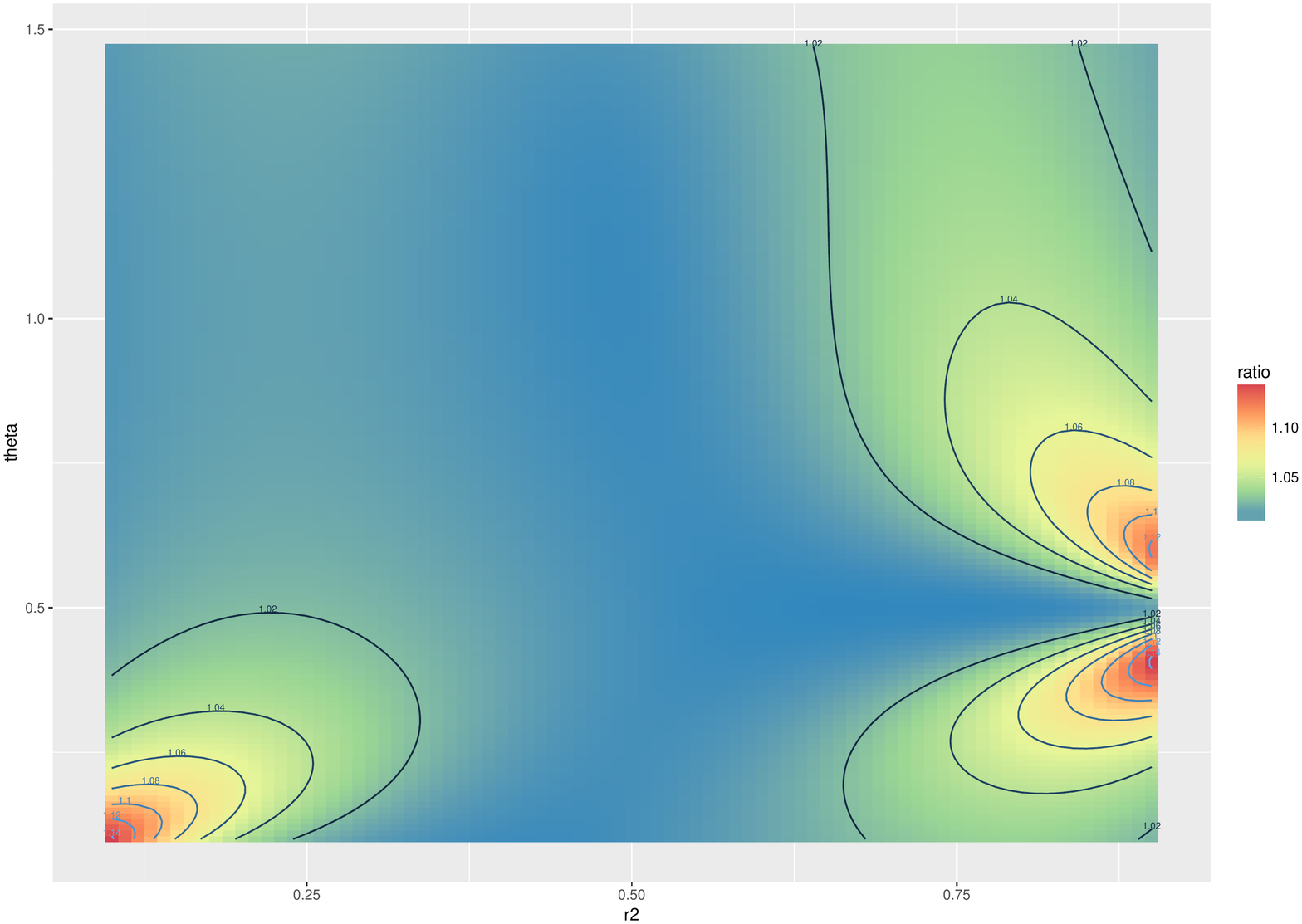}} 
\caption{The ratios $\mathrm{MSE}(\hat{\bB}^{(S)})/\mathrm{MSE}(\hat{\bB}^{(N)})$ and $\mathrm{MSE}(\hat{\bB}^{(S)})/\mathrm{MSE}(\hat{\bB}^{(M)})$ for values of 
$r_2 \in [0.1, 0.9]$ and $\theta \in [0.1, \tfrac{\pi}{2} - 0.1]$ in a $3$-block rank $2$ SBM. 
The labeled lines in each plot are the contour lines for the $\mathrm{MSE}$ ratios. 
The $\mathrm{MSE}(\hat{\bB}^{(S)})$ is computed using the bias-adjusted estimates $\{\hat{\mathbf{B}}^{(S)}_{ij} - \hat{\theta}_{ij}\}_{i \leq j}$ (see Corollary~\ref{cor:in_practice}). 
}
\label{fig:ratio-plot2}
\end{figure}
Plots of the MSE (equivalently the sum of variances
) for the bias-adjusted estimates $\hat{\mathbf{B}}^{(S)}$ against the MSE (equivalently the sum of variances) of the naive estimates $\hat{\mathbf{B}}^{(N)}$ and the true MLE estimates $\hat{\mathbf{B}}^{(M)}$ for a subset of the parameters of a $3$-blocks SBM with $\rho_n \equiv 1$ given in Figure~\ref{fig:ratio-plot2}. In Figure~\ref{fig:ratio-plot2}, we fix $\bm{\pi} = (1/3,1/3,1/3)$, $r_3 = 0.7$, $\gamma = 0.5$, $r_1 = 1- r_2$, letting $r_2$ and $\theta$ vary in the intervals $[0.1,0.9]$ and $[0.1, \tfrac{\pi}{2} - 0.1]$, respectively. Once again, $\hat{\mathbf{B}}^{(S)}$ has smaller mean squared error than $\hat{\mathbf{B}}^{(N)}$ over the whole range of $r_2$ and $\theta$, and has mean squared error almost as small as that of $\hat{\mathbf{B}}^{(M)}$ for a large range of $r_2$ and $\theta$.

\end{remark}

We emphasize that the rank assumptions placed on $\mathbf{B}$ in the previous two examples are natural assumptions, i.e., there is no primafacie reason why $\mathbf{B}$ needs to be invertible, and hence procedures that can both estimate $\mathrm{rk}(\mathbf{B})$ and incorporate it in the subsequent estimation of $\mathbf{B}$ are equally flexible and generally more efficient.  This is in contrast to other potentially more restrictive assumptions such as assuming that $\mathbf{B}$ is of the form $q \bm{1} \bm{1}^{\top} + (p - q) \mathbf{I}$ for $p > q$ (i.e., the planted-partitions model). Indeed, a $K$-block SBM from the planted-partitions model is parametrized by two parameters, irrespective of $K$ and as such the three estimators considered in this paper are provably sub-optimal for estimating the parameters of the planted partitions model. 
\section{Discussions}
Theorem~\ref{THM:GEN_D} and Theorem~\ref{THM:GEN_D_SPARSE} were presented in the context wherein the vertices to block assignments $\bm{\tau}$ are assumed known. For unknown $\bm{\tau}$, Lemma~\ref{LEM:PERFECT} (presented below) implies that $\hat{\bm{\tau}}$ obtained using $K$-means (or Gaussian mixture modeling) on the rows of $\hat{\mathbf{U}}$ is an exact recovery of $\bm{\tau}$, provided that $n \rho_n \ = \omega(\log n)$. The lemma implies Corollary~\ref{cor:in_practice} showing that we can replace the quantities $\theta_{k \ell}$ and $\tilde{\theta}_{k \ell}$ in Eq.~\eqref{eq:sbm_normal2} and Eq.~\eqref{eq:sbm_normal3} of Theorem~\ref{THM:GEN_D} and Theorem~\ref{THM:GEN_D_SPARSE} by consistent estimates $\hat{\theta}_{k \ell}$ without changing the resulting limiting distribution. We emphasize that it is essential for Corollary~\ref{cor:in_practice} that $\hat{\bm{\tau}}$ is an exact recovery of $\bm{\tau}$ in order for the limiting distributions in Eq.~\eqref{eq:sbm_normal2} and Eq.~\eqref{eq:sbm_normal3} to remain valid when $\hat{\theta}_{k \ell}$ is substituted for $\theta_{k \ell}$ and $\tilde{\theta}_{k \ell}$. Indeed, if there is even a single vertex that is mis-clustered by $\hat{\bm{\tau}}$, then $\hat{\theta}_{k \ell}$ as defined will introduce an additional (random) bias term in the limiting distribution of Eq.~\eqref{eq:sbm2_inpractice}. 

\begin{remark}
For ease of exposition, bounds in this paper are often written as holding ``with high probability''. A random variable $\xi \in \mathbb{R}$ is $O_{\mathbb{P}}(f(n))$ if, for any positive constant $c > 0$ there exists a $n_0 \in \mathbb{N}$ and a constant $C > 0$ (both of which possibly depend on $c$) such that for all $n \geq n_0$, $|\xi| \leq C f(n)$ with probability at least $1 - n^{-c}$; moreover, a random variable $\xi \in \mathbb{R}$ is $o_{\mathbb{P}}(f(n))$ if for any positive constant $c > 0$ and any $\epsilon > 0$ there exists a $n_0 \in \mathbb{N}$ such that for all $n \geq n_0$, $|\xi| \leq \epsilon f(n)$ with probability at least $1 - n^{-c}$. Similarly, when $\xi$ is a random vector in $\mathbb{R}^{d}$ or a random matrix in $\mathbb{R}^{d_1 \times d_2}$, $\xi = O_{\mathbb{P}}(f(n))$ or $\xi = o_{\mathbb{P}}(f(n))$ if $\|\xi\| = O_{\mathbb{P}}(f(n))$  or $\|\xi\| = o_{\mathbb{P}}(f(n))$, respectively. Here $\|x\|$ denotes the Euclidean norm of $x$ when $x$ is a vector and the spectral norm of $x$ when $x$ is a matrix. We write $\xi = \zeta + O_{\mathbb{P}}(f(n))$ or $\xi = \zeta + o_{\mathbb{P}}(f(n))$ if $\xi - \zeta = O_{\mathbb{P}}(f(n))$ or $\xi - \zeta = o_{\mathbb{P}}(f(n))$, respectively. 
\end{remark}

\begin{lemma}
\label{LEM:PERFECT}
Let $(\mathbf{A}_n, \mathbf{X}_n) \sim \mathrm{GRDPG}_{p,q}(F)$ be a generalized random dot product graph on $n$ vertices
with sparsity factor $\rho_n$. 
Let $\hat{\mathbf{U}}_n(i)$ and $\mathbf{U}_n(i)$ be the $i$-th row of $\hat{\mathbf{U}}_n$ and $\mathbf{U}_n$, respectively. Here $\hat{\mathbf{U}}_n$ and $\mathbf{U}_n$ are the eigenvectors of $\mathbf{A}_n$ and 
$\mathbf{X}_n \mathbf{X}_n^{\top}$ corresponding to the $p + q$ largest eigenvalues (in modulus) of $\mathbf{A}_n$ and $\mathbf{X}_n \mathbf{X}_n^{\top}$. Then there exist a (universal) constant $c > 0$ and a $d \times d$ 
orthogonal matrix $\mathbf{W}_n$ such that, for $n \rho_n = \omega(\log^{2c}(n)$, 
\begin{equation}
\label{eq:perfect} 
\max_{i \in [n]} \|\mathbf{W}_n \hat{\mathbf{U}}_n(i) - \mathbf{U}_n(i)\| = O_{\mathbb{P}}\Bigl(\frac{\log^{c}{n}}{n \sqrt{\rho_n}}\Bigr).
\end{equation}
\end{lemma}
The proof of Lemma~\ref{LEM:PERFECT} is given in the appendix.
If $\mathbf{A}_n \sim \mathrm{SBM}(\mathbf{B}, \bm{\pi})$ with sparsity factor $\rho_n$ and $\mathbf{B}$ is $K \times K$, then the rows of $\mathbf{U}_n$ take on at most $K$ possible distinct values. Moreover, for any vertices $i$ and $j$ with $\tau_i \not = \tau_j$, $\|\mathbf{U}_n(i) - \mathbf{U}_n(j)\| \geq Cn^{-1/2}$ for some constant $C$ depending only on $\mathbf{B}$. Now if $n \rho_n = \omega(\log^{2c}(n))$, then Lemma~\ref{LEM:PERFECT} implies, for sufficiently large $n$,
$$\|\mathbf{W}_n \hat{\mathbf{U}}_n(i) - \mathbf{U}_n(i)\| < \min_{j \colon \tau_j \not = i}\|\mathbf{W}_n \hat{\mathbf{U}}_n(i) - \mathbf{U}_n(j)\|; \quad \text{for all $i \in [n]$}.$$
Hence, since $\mathbf{W}_n$ is an orthogonal matrix, $K$-means clustering of the rows of $\hat{\mathbf{U}}_n$ yield an assignment $\hat{\bm{\tau}}$ that is indeed, up to a permutation of the block labels, an exact recovery of $\bm{\tau}$ as $n \rightarrow \infty$. We note that Lemma~\ref{LEM:PERFECT} is an extension of our earlier results on bounding the perturbation $\hat{\mathbf{U}}_n - \mathbf{U}_n \mathbf{W}$ using the $2 \to \infty$ matrix norm \cite{perfect,lyzinski15_HSBM,ctp_2_to_infty}. Lemma~\ref{LEM:PERFECT} is very similar in flavor to a few recent results by other researchers \cite{abbe2,mao_sarkar,belkin_2_infty} where eigenvector perturbations of $\mathbf{A}_n$ (compared to the eigenvectors of $\mathbf{X}_n \mathbf{X}_n^{\top}$)
in the $\ell_{\infty}$ norm is established in the regime where $n \rho_n = \omega(\log^{c}(n))$ for some constant $c > 0$.  
\begin{corollary}
\label{cor:in_practice}
Assume the setting and notations of Theorem~\ref{THM:GEN_D}. Assume $K$ known, let $\hat{\bm{\tau}} \colon [n] \mapsto [K]$ be the vertex to cluster assignments when the rows of $\hat{\mathbf{U}}$ are clustered into $K$ clusters. For $k \in [K]$, let $\hat{\bm{s}}_k \in \{0,1\}^{n}$ where the $i$-th entry of $\hat{\bm{s}}_k$ is $1$ if $\hat{\tau}_i = k$ and $0$ otherwise. Let $\hat{n}_k = |\{i \colon \hat{\tau}_i = k \}|$ and let $\hat{\pi}_k = \tfrac{\hat{n}_k}{n}$. For $k \in [K]$, let $\hat{\nu}_k = \tfrac{1}{\hat{n}_k} \hat{\bm{s}}_k^{\top} \hat{\bU} \hat{\bm{\Lambda}}^{1/2}$, 
let $\hat{\mathbf{B}}_{k \ell} = \hat{\mathbf{B}}_{k \ell}^{(S)} = \hat{\nu}_k^{\top} \mathbf{I}_{p,q} \hat{\nu}_{\ell}$, and
let $\hat{\Delta} = \sum_{k} \hat{\pi}_k \hat{\nu}_k \hat{\nu}_k^{\top}$. For $k \in [K]$ and $\ell \in [K]$, let $\hat{\theta}_{k \ell}$ be given by
\begin{equation}
 \label{eq:hat.mu_kl}
 \begin{split}
\theta_{k\ell} &= \sum_{r=1}^{K} \hat{\pi}_r \bigl(\hat{\mathbf{B}}_{kr} (1 - \hat{\mathbf{B}}_{kr}) + \hat{\mathbf{B}}_{\ell r}(1 - \hat{\mathbf{B}}_{\ell r})\bigr)\hat{\nu}_k^{\top} \hat{\Delta}^{-1} \mathbf{I}_{p,q} \hat{\Delta}^{-1} \nu_{\ell} \\
& - \sum_{r=1}^{K} \sum_{s=1}^{K} \hat{\pi}_r \hat{\pi}_s \hat{\mathbf{B}}_{rs} (1 - \hat{\mathbf{B}}_{sr}) \hat{\nu}_s^{\top} \hat{\Delta}^{-1} \mathbf{I}_{p,q} \hat{\Delta}^{-1} 
(\hat{\nu}_{\ell} \hat{\nu}_k^{\top} + \hat{\nu}_k \hat{\nu}_{\ell}^{\top}) \hat{\Delta}^{-1} \hat{\nu}_s.
 \end{split}
 \end{equation}
 Then there exists a (sequence of) permutation(s) $\psi \equiv \psi_n$ on $[K]$ such that for any $k \in [K]$ and $\ell \in [K]$,
 \begin{equation}
 \label{eq:sbm2_inpractice}
 n(\hat{\mathbf{B}}^{(S)}_{\psi(k), \psi(\ell)} - \mathbf{B}_{k \ell} - \tfrac{\hat{\theta}_{k \ell}}{n}) \overset{\mathrm{d}}{\longrightarrow} \mathcal{N}(0, \sigma_{k \ell}^2)
 \end{equation}
as $n \rightarrow \infty$.
\end{corollary}
An almost identical result hold in the setting when $\rho_n \rightarrow 0$. More specifically, assume the setting and notations of Theorem~\ref{THM:GEN_D_SPARSE} and 
let $\hat{\nu}_k$, $\hat{\Delta}$ and $\hat{\mathbf{B}} = \hat{\mathbf{B}}^{(S)}$ be as defined in Corllary~\ref{cor:in_practice}. Now let $\hat{\theta}_{k \ell}$ be given by
\begin{equation}
 \label{eq:hat.mu_kl2}
 \begin{split}
\hat{\theta}_{k\ell} &= \sum_{r=1}^{K} \hat{\pi}_r \bigl(\hat{\mathbf{B}}_{kr} + \hat{\mathbf{B}}_{\ell r} \bigr)\hat{\nu}_k^{\top} \hat{\Delta}^{-1} \mathbf{I}_{p,q} \hat{\Delta}^{-1} \nu_{\ell} \\
& - \sum_{r=1}^{K} \sum_{s=1}^{K} \hat{\pi}_r \hat{\pi}_s \hat{\mathbf{B}}_{rs} \hat{\nu}_s^{\top} \hat{\Delta}^{-1} \mathbf{I}_{p,q} \hat{\Delta}^{-1} 
(\hat{\nu}_{\ell} \hat{\nu}_k^{\top} + \hat{\nu}_k \hat{\nu}_{\ell}^{\top}) \hat{\Delta}^{-1} \hat{\nu}_s.
 \end{split}
 \end{equation}
 Then there exists a (sequence of) permutation(s) $\psi \equiv \psi_n$ on $[K]$ such that for any $k \in [K]$ and $\ell \in [K]$,
 \begin{equation}
 n \rho_n^{1/2} (\hat{\mathbf{B}}_{\psi(k), \psi(\ell)} - \mathbf{B}_{k \ell} - \tfrac{\hat{\theta}_{k \ell}}{n \rho_n}) \overset{\mathrm{d}}{\longrightarrow} \mathcal{N}(0, \tilde{\sigma}_{k \ell}^2)
 \end{equation}
as $n \rightarrow \infty$, $\rho_n \rightarrow 0$ and $n \rho_n = \omega(\sqrt{n}).$

Finally, we provide some justification on the necessity of the assumption $n \rho_n = \omega(\sqrt{n})$ in the statement of Theorem~\ref{THM:GEN_D_SPARSE}, even though Lemma~\ref{LEM:PERFECT} implies that $\hat{\bm{\tau}}$ is an exact recovery of $\bm{\tau}$ for $n \rho_n = \omega(\log^{2c}(n))$. 
Consider the case of $\mathbf{A}$ being an Erd\H{o}s-R\'{e}nyi graph on $n$ vertices with edge probability $p$. 
The estimate $\hat{p}$ obtained from the spectral embedding in this setting is $\tfrac{1}{n^2} \hat{\lambda} (\bm{1}^{\top} \hat{\bm{u}})^2$ where $\hat{\lambda}$ is the largest eigenvalue of $\mathbf{A}$, $\bm{1}$ is the all ones vector, and $\hat{\bm{u}}$ is the associated (unit-norm) eigenvector. Let $\bm{e} = n^{-1/2} \bm{1}$.
We then have
\begin{equation*}
\begin{split}
n (\hat{p} - p) &= \tfrac{1}{n} \hat{\lambda} (\bm{1}^{\top} \hat{\bm{u}})^2 - np = \hat{\lambda} \bigl((\bm{e}^{\top} \hat{\bm{u}})^2 - 1\bigr) + \hat{\lambda} - np.
\end{split}
\end{equation*}
When $p$ remains constant as $n$ changes, then the results of \cite{furedi1981eigenvalues} implies $\bm{e}^{\top} \hat{\bm{u}} = 1 - \tfrac{1-p}{2np} + O_{\mathbb{P}}(n^{-3/2})$ and $\hat{\lambda} - np = \bm{e}^{\top} (\mathbf{A} - \mathbb{E}[\mathbf{A}]) \bm{e} + (1-p) + O_{\mathbb{P}}(n^{-1/2})$, from which we infer 
\begin{equation*}
\begin{split}
n (\hat{p} - p) &= - (1-p) \tfrac{\hat{\lambda}}{np} + \bm{e}^{\top} (\mathbf{A} - \mathbb{E}[\mathbf{A}]) \bm{e} + (1-p) + O_{\mathbb{P}}(n^{-1/2}) \\ &= \bm{e}^{\top} (\mathbf{A} - \mathbb{E}[\mathbf{A}]) \bm{e} + O_{\mathbb{P}}(n^{-1/2}) \overset{\mathrm{d}}{\longrightarrow} \mathcal{N}(0, 2p(1-p))
\end{split}
\end{equation*}
since $\bm{e}^{\top} (\mathbf{A} - \mathbb{E}[\mathbf{A}]) \bm{e}$ is a sum of $n(n+1)/2$ independent mean $0$ random variables with variance $p(1-p)$.
On the other hand, if $p \rightarrow 0$ as $n$ increases, then Theorem~6.2 of \cite{erdos} (more specifically Eq.~(6.9) and Eq.~(6.26) of \cite{erdos}) implies
\begin{gather}
\label{eq:et_hatu}
\bm{e}^{\top} \hat{\bm{u}} = 1 - \tfrac{1-p}{2np} + O_{\mathbb{P}}\bigl((np)^{-3/2} + \tfrac{\log^{c}n}{n \sqrt{p}}\bigr)
\end{gather}
and 
\begin{equation}
\label{eq:lambda_ER}
\begin{split}
 \hat{\lambda} - np &= \bm{e}^{\top} (\mathbf{A} - \mathbf{E}[\mathbf{A}]) \bm{e} + \tfrac{\bm{e}^{\top} (\mathbf{A} - \mathbb{E}[\mathbf{A}])^2 \bm{e}}{np} + O_{\mathbb{P}}\bigl((np)^{-1} + \tfrac{\log^{c}{n}}{n \sqrt{p}}) \\
 &= \bm{e}^{\top} (\mathbf{A} - \mathbf{E}[\mathbf{A}]) \bm{e} + (1 - p) + O_{\mathbb{P}}\bigl((np)^{-1} + \tfrac{\log^{c}{n}}{n \sqrt{p}}).
 \end{split}
\end{equation}
The second equality in Eq.~\eqref{eq:lambda_ER} follows from Lemma~6.5 of \cite{erdos} which states that $\bm{e}^{\top} (\mathbf{A} - \mathbf{E}[\mathbf{A}])^{k} \bm{e} = \bm{e}^{\top} \mathbb{E}[(\mathbf{A} - \mathbb{E}[\mathbf{A}])^{k}] \bm{e} + O_{\mathbb{P}}(\tfrac{(np)^{k/2} \log^{kc}(n)}{\sqrt{n}})$ for some universal constant $c > 0$ provided that $n p = \omega(\log{n})$. Hence
\begin{equation}
n(\hat{p} - p) = - (1-p) \tfrac{\hat{\lambda}}{np} + \bm{e}^{\top} (\mathbf{A} - \mathbf{E}[\mathbf{A}]) \bm{e} + (1-p) + O_{\mathbb{P}}((np)^{-1/2}).
\end{equation}
Once again $\bm{e}^{\top} (\mathbf{A} - \mathbf{E}[\mathbf{A}]) \bm{e}$ is a sum of $n(n+1)/2$ independent mean $0$ random variables with variance $p(1-p)$, but since $p \rightarrow 0$, the individual variance also vanishes as $n \rightarrow \infty$. In order to obtain a non-degenerate limiting distribution for $\bm{e}^{\top} (\mathbf{A} - \mathbf{E}[\mathbf{A}]) \bm{e}$, it is necessary that we consider $p^{-1/2} \bm{e}^{\top} (\mathbf{A} - \mathbf{E}[\mathbf{A}]) \bm{e}$. This, however, lead to non-trivial technical difficulties. In particular, 
\begin{equation*}
\begin{split}
n p^{-1/2}(\hat{p} - p) &= p^{-1/2} \bm{e}^{\top} (\mathbf{A} - \mathbf{E}[\mathbf{A}]) \bm{e} + (1-p)(\tfrac{\hat{\lambda} - np}{np^{3/2}}) + O_{\mathbb{P}}(n^{-1/2} p^{-1}) \\
&= p^{-1/2} \bm{e}^{\top} (\mathbf{A} - \mathbf{E}[\mathbf{A}]) \bm{e} \bigl(1 + \tfrac{1-p}{np}\bigr) + O_{\mathbb{P}}(n^{-1/2} p^{-1})
\end{split}
\end{equation*}
upon iterating the term $(\hat{\lambda} - np)$. To guarantee that $O_{\mathbb{P}}(n^{-1/2} p^{-1})$ vanishes in the above expression, 
it might be necessary to require $n p = \omega(\sqrt{n})$. That is to say, the expansions for $\bm{e}^{\top} \hat{\bm{u}}$ and $\hat{\lambda} - np$ in Eq.~\eqref{eq:et_hatu} and Eq.~\eqref{eq:lambda_ER} is not sufficiently refined. 

We surmise that to extend Theorem~\ref{THM:GEN_D_SPARSE}, even in the context of Erd\H{o}s-R\'{e}nyi graphs, to the setting wherein $n p = o(\sqrt{n})$, it is necessary to consider higher order expansion for $\bm{e}^{\top} \hat{\bm{u}}$ and $\hat{\lambda} - np$. But this necessitates evaluating $\bm{e}^{\top} \mathbb{E}[(\mathbf{A} - \mathbb{E}[\mathbf{A}])^{k}] \bm{e}$ for $k \geq 3$, a highly non-trivial task; in particular $n p = \omega(\log{n})$ potentially require evaluating $\mathbb{E}[\bm{e}^{\top}(\mathbf{A} - \mathbb{E}[\mathbf{A}])^{k} \bm{e}]$ for $k = O(\log{n})$. In a slightly related vein, 
\cite{zongmingma} evaluates $\mathrm{tr}[\mathbb{E}[(\mathbf{A} - \mathbb{E}[\mathbf{A}])^{k}$ in the case of Erd\H{o}s-R\'{e}nyi graphs and two-blocks planted partition SBM graphs. 

Exact recovery of $\bm{\tau}$ via $\hat{\bm{\tau}}$ is therefore not sufficient to guarantee control of $\hat{\mathbf{B}}^{(S)}_{k \ell} - \mathbf{B}_{k \ell} = \bm{s}_k^{\top} (\hat{\bU} \hat{\bm{\Lambda}} \hat{\bU}^{\top} - \mathbb{E}[\mathbf{A}]) \bm{s}_{\ell}$. In essence, as $\rho_n \rightarrow 0$, the bias incurred by the low-rank approximation $
\hat{\bU} \hat{\bm{\Lambda}} \hat{\bU}^{\top}$ of $\mathbf{A}$ overwhelms the reduction in variance resulting from the low-rank approximation. 



\appendix
\section{Proof of Theorem~\ref{THM:GEN_D} and Theorem~\ref{THM:GEN_D_SPARSE}}
We first provide an outline of the main steps in the proof of Theorem~\ref{THM:GEN_D} and Theorem~\ref{THM:GEN_D_SPARSE}.
We derive Eq.~\eqref{eq:sbm_normal3} (and analogously Eq.~\eqref{eq:sbm_normal2}) by considering the following decomposition of $(\hat{\mathbf{B}}^{(S)}_{k\ell} - \mathbf{B}_{k \ell})$ 
\begin{equation}
\label{eq:decomp_Bhat-B}
\begin{split}
n \rho_n^{1/2}(\hat{\mathbf{B}}^{(S)}_{k\ell} - \mathbf{B}_{k \ell}) & = \tfrac{n \rho_n^{1/2}}{\hat{n}_k \hat{n}_{\ell} \rho} \hat{\bm{s}}_k^{\top} \hat{\mathbf{U}} \hat{\bm{\Lambda}} \hat{\mathbf{U}}^{\top} \hat{\bm{s}}_{\ell} - \tfrac{n \rho_n^{1/2}}{n_{k} n_{\ell} \rho} \bm{s}_k^{\top}\mathbb{E}[\mathbf{A}] \bm{s}_{\ell} \\ &=
\tfrac{n \rho_n^{-1/2}}{n_k n_{\ell}} \bm{s}_k^{\top} (\hat{\mathbf{U}} \hat{\bm{\Lambda}} \hat{\mathbf{U}}^{\top} - \mathbf{U} \mathbf{U}^{\top} \hat{\mathbf{U}}
\hat{\bLam} \hat{\bU}^{\top} \bU \bU^{\top}) \bm{s}_{\ell} \\ & +
\tfrac{n \rho_n^{-1/2}}{n_k n_{\ell}} \bm{s}_k^{\top} \bU (\bU^{\top} \hat{\bU} \hat{\bLam} - \bLam \bU^{\top} \hat{\bU}) 
\hat{\bU}^{\top} \bU \bU^{\top} \bm{s}_{\ell} \\
&+  \tfrac{n \rho_n^{-1/2}}{n_k n_{\ell}} \bm{s}_k^{\top} \bU \bLam (\bU^{\top} \hat{\bU} \hat{\bU}^{\top} \bU - \mathbf{I}) \bU^{\top} \bm{s}_{\ell}.
\end{split}
\end{equation}
Our proof proceeds by writing each term on the right hand side of Eq.~\eqref{eq:decomp_Bhat-B} as, when conditioned on $\mathbf{P}$, linear combinations of the 
independent random variables $\{\mathbf{A}_{ij} - \mathbf{P}_{ij}\}_{i \leq j}$ and residual terms of smaller order. More specifically, letting $\mathbf{E} = \mathbf{A} - \mathbf{P}$, $\bm{\Pi}_{\mathbf{U}} = \mathbf{U} \mathbf{U}^{\top}$, $\bm{\Pi}^{\perp}_{\mathbf{U}} = \mathbf{I} - \bm{\Pi}_{\mathbf{U}}$ and $\mathbf{P}^{\dagger} = \mathbf{U} \bm{\Lambda}^{-1} \mathbf{U}^{\top}$ the Moore-Penrose pseudoinverse of $\mathbf{P}$, we show that 
\begin{equation}
\begin{split}
\label{eq:xi_kl1}
\xi_{k \ell}^{(1)} & := \tfrac{n \rho_n^{-1/2}}{n_k n_{\ell}} \bm{s}_k^{\top} (\hat{\mathbf{U}} \hat{\bm{\Lambda}} \hat{\mathbf{U}}^{\top} - \mathbf{U} \mathbf{U}^{\top} \hat{\mathbf{U}}
\hat{\bLam} \hat{\bU}^{\top} \bU \bU^{\top}) \bm{s}_{\ell}
\\ &= \tfrac{n \rho_n^{-1/2}}{n_{k} n_{\ell}} \bm{s}_k^{\top} \bm{\Pi}_{\bU}^{\perp} \bigl(\mathbf{E} \bm{\Pi}_{\bU} + \mathbf{E}^2 \mathbf{P}^{\dagger} \bigr) \bm{s}_{\ell} \\ &+ \tfrac{n \rho_n^{-1/2}}{n_{k} n_{\ell}} \bm{s}_{\ell}^{\top} \bm{\Pi}_{\bU}^{\perp} \bigl(\mathbf{E} \bm{\Pi}_{\bU}  + \mathbf{E}^2 \mathbf{P}^{\dagger} \bigr) \bm{s}_{k} + O_{\mathbb{P}}(n^{-1/2} \rho_n^{-1}),
\end{split}
\end{equation}
\begin{equation}
\begin{split}
\label{eq:xi_kl3}
\xi_{k \ell}^{(3)} & := 
\tfrac{n \rho_n^{-1/2}}{n_k n_{\ell}} \bm{s}_k^{\top} \bU \bLam (\bU^{\top} \hat{\bU} \hat{\bU}^{\top} \bU - \mathbf{I}) \bU^{\top} \bm{s}_{\ell}.
\\ &= -\tfrac{n \rho_n^{-1/2}}{n_{k} n_{\ell}} \bm{s}_k^{\top} \bm{\Pi}_{\bU} \mathbf{E}^2 \mathbf{P}^{\dagger} \bm{s}_{\ell} + 
O_{\mathbb{P}}(n^{-1/2} \rho_n^{-1}),
\end{split}
\end{equation}
\begin{equation}
\begin{split}
\label{eq:xi_kl2}
\xi_{k \ell}^{(2)} & := 
\tfrac{n \rho_n^{-1/2}}{n_k n_{\ell}} \bm{s}_k^{\top} \bU (\bU^{\top} \hat{\bU} \hat{\bLam} - \bLam \bU^{\top} \hat{\bU}) 
\hat{\bU}^{\top} \bU \bU^{\top} \bm{s}_{\ell} \\ &=
\tfrac{n \rho_n^{-1/2}}{n_{k} n_{\ell}} \bm{s}_{k}^{\top} \bm{\Pi}_{\bU} \mathbf{E} \bm{\Pi}_{\bU} \bm{s}_{\ell} - \xi_{kl}^{(3)} + 
O_{\mathbb{P}}(n^{-1/2} \rho_n^{-1}).
\end{split}
\end{equation}
The above expressions for $\xi_{k \ell}^{(1)}, \xi_{kl}^{(2)}$ and $\xi_{kl}^{(3)}$ implies
\begin{equation}
\label{eq:decomp_Bhat-B3}
\begin{split}
\tfrac{n \rho_n^{1/2}}{n_{k} n_{\ell}} (\hat{\mathbf{B}}_{k\ell}^{(S)} - \mathbf{B}_{k \ell}) &= 
\tfrac{n \rho_n^{-1/2}}{n_{k} n_{\ell}} \bigl(\bm{s}_k^{\top} \mathbf{E} \bm{\Pi}_{\bU} \bm{s}_{\ell} + \bm{s}_{\ell}^{\top} \bm{\Pi}_{\bU}^{\perp} \mathbf{E} \bm{\Pi}_{\bU} \bm{s}_{k} \bigr) \\ &+
\tfrac{n \rho_n^{-1/2}}{n_{k} n_{\ell}} \bigl(\bm{s}_{k}^{\top} \bm{\Pi}_{\bU}^{\perp} \mathbf{E}^{2} \mathbf{P}^{\dagger} \bm{s}_{\ell} + \bm{s}_{\ell}^{\top} \bm{\Pi}_{\bU}^{\perp} \mathbf{E}^{2} \mathbf{P}^{\dagger} \bm{s}_{k} \bigr) \\ &
 + O_{\mathbb{P}}(n^{-1/2} \rho_n^{-1}).
\end{split}
\end{equation}
We complete the proof of by showing that 
\begin{equation}
\label{eq:Zkl}
\begin{split}
Z_{kl} &:= \tfrac{n \rho_n^{-1/2}}{n_{k} n_{\ell}} \Bigl( \bm{s}_k^{\top} \mathbf{E} \bm{\Pi}_{\bU} \bm{s}_{\ell} + 
\bm{s}_{\ell}^{\top} \bm{\Pi}_{\bU}^{\perp} \mathbf{E} \bm{\Pi}_{\bU} \bm{s}_{k} \Bigr) \\ &
= \tfrac{n \rho_n^{-1/2}}{n_{k} n_{\ell}} \mathrm{tr} \,\, \mathbf{E} (\bm{\Pi}_{\bU} \bm{s}_{\ell} \bm{s}_k^{\top} + \bm{\Pi}_{\bU} \bm{s}_{k} \bm{s}_{\ell}^{\top} \bm{\Pi}_{\bU}^{\perp} ) \\
&= \tfrac{n \rho_n^{-1/2}}{n_{k} n_{\ell}} \mathrm{tr} \,\, \mathbf{E} (\bm{\Pi}_{\bU} \bm{s}_{\ell} \bm{s}_k^{\top} + \bm{\Pi}_{\bU} \bm{s}_{k} \bm{s}_{\ell}^{\top}  - \bm{\Pi}_{\bU} \bm{s}_{k} \bm{s}_{\ell}^{\top} \bm{\Pi}_{\bU}) 
\end{split}
\end{equation}
converges to a normally distributed random variable, and that
\begin{equation} 
\label{eq:aux_theta_kl1}
\tfrac{n}{n_{k} n_{\ell}} \bigl(\bm{s}_{k}^{\top} \bm{\Pi}_{\bU}^{\perp} \mathbf{E}^{2} \mathbf{P}^{\dagger}  + \bm{s}_{\ell}^{\top} \bm{\Pi}_{\bU}^{\perp} 
\mathbf{E}^{2} \mathbf{P}^{\dagger} \bm{s}_{k} \bigr) \overset{\mathrm{a.s.}}{\longrightarrow} \begin{cases} \theta_{k \ell} & \text{if $\rho_n \equiv 1$} \\ \tilde{\theta}_{k \ell} & \text{if $\rho_n \rightarrow 0$} 
\end{cases}
\end{equation}
as $n \rightarrow \infty$. Note the difference in scaling for the convergence of $Z_{kl}$ (scaling by $\tfrac{n \rho_n^{-1/2}}{n_{k} n_{\ell}}$) and the scaling in Eq.~\eqref{eq:aux_theta_kl1} (scaling by $\tfrac{n}{n_{k} n_{\ell}}$).

We now provide the necessary details for the proof sketch outlined above. 
We shall repeatedly make use of the following concentration bounds for $\|\mathbf{A} - \mathbf{P}\|$ and related quantities. We consolidated these bounds in the following lemma. 
\begin{lemma}
\label{lem:A-P}
Let $\mathbf{A} \sim \mathrm{GRDPG}_{p,q}(F)$ be a generalized random dot product graph on $n$ vertices with sparsity factor $\rho_n$. 
Suppose $n \rho_n = \omega(\log^{4}(n))$. Then
  \begin{gather}
  \label{eq:A-P}
  \|\mathbf{A} - \mathbf{P}\| = O_{\mathbb{P}}((n \rho_n)^{1/2}) \\
  \label{eq:UU-UhatUhat}
  \|\bU \bU^{\top} - \hat{\bU} \hat{\bU}^{\top} \| = 
  O_{\mathbb{P}}((n \rho_n)^{-1/2}) \\
  \label{eq:I-UU_Uhat}
  \|(\mathbf{I} - \bU \bU^{\top}) \hat{\bU}\| = O_{\mathbb{P}}((n \rho_n)^{-1/2}).
  \end{gather}
  In addition, there exists an orthogonal matrix $\mathbf{W}$ such that
  \begin{equation}
  \label{eq:UTU_hat-W}
  \|\bU^{\top} \hat{\bU} - \mathbf{W} \| = O_{\mathbb{P}}((n \rho_n)^{-1})
  \end{equation}
\end{lemma}
The bound for $\|\mathbf{A} - \mathbf{P}\|$ in Eq.~\eqref{eq:A-P} is due to \cite{lu13:_spect}. For ease of exposition, we have stated Eq.~\eqref{eq:A-P} in the context of generalized random dot product graph and hence the upper bound is given in terms of the factor $n \rho_n$; the original bound holds for the more general inhomogeneous random graphs model where the upper bound is now given in terms of $\sqrt{\delta}$ where $\delta = \max_{i} \sum_{j} \mathbf{P}_{ij}$ is the maximum expected degree.
Similar upper bounds can be found in \cite{oliveira2009concentration,tropp,rinaldo_2013} with slightly different assumptions on $\mathbf{P}$. The bound for $\|\bU \bU^{\top} - \hat{\bU} \hat{\bU}^{\top}\|$ and $\|(\mathbf{I} - \bU \bU^{\top}) \hat{\bU}\|$ then follows from Eq.~\eqref{eq:A-P} and the Davis-Kahan theorem \cite{davis70,stewart90:_matrix,samworth}. Eq.~\eqref{eq:UTU_hat-W} follows from Eq.~\eqref{eq:UU-UhatUhat} via the following argument.
 Let $\sigma_1, \sigma_2, \dots, \sigma_d$ denote the singular values of
  $\bU^{\top} \hat{\bU}$. Then $\sigma_i = \cos(\theta_i)$ where
  the $\theta_i$ are the principal angles between the subspaces
  spanned by $\bU^{\top} \hat{\bU}$. Eq.~\eqref{eq:UU-UhatUhat} implies 
  \begin{equation*}
    \| \bU \bU^{\top} - \hat{\bU} \hat{\bU}^{\top} \| =
    \max_{i} | \sin(\theta_i) | = O_{\mathbb{P}}((n \rho_n)^{-1/2}).
  \end{equation*}
  Let $\mathbf{W}_1 \bm{\Sigma} \mathbf{W}_2^{\top}$ be the singular value decomposition of $\bU^{\top} \hat{\bU}$ and let $\mathbf{W} = \mathbf{W}_1 \mathbf{W}_2^{\top}$. We then have
  \begin{equation*}
    \begin{split}
    \|\bU^{\top} \hat{\bU} - \mathbf{W} \|_{F} = \|\bm{\Sigma} -
    \mathbf{I} \|_{F} &= \Bigl(\sum_{i=1}^{d} (1 - \sigma_i)^2\Bigr)^{1/2}  \leq \sum_{i=1}^{d} (1 -
    \sigma_i^{2}) = \sum_{i=1}^{d} \sin^{2}(\theta_i).
  \end{split}
  \end{equation*}
  Hence $\|\bU^{\top} \hat{\bU} - \mathbf{W} \|_{F}
     = O_{\mathbb{P}}((n \rho_n)^{-1})
  $ as desired.

We shall also repeatedly make use of a von-Neumann expansion for $\hat{\bU}$. More specifically, from $\mathbf{A} \hat{\bU} = \hat{\bU} \bm{\Lambda}$, we have
$$ \hat{\bU} \hat{\bm{\Lambda}} - (\mathbf{A} - \mathbf{P}) \hat{\bU} = \mathbf{P} \hat{\bU} $$
which is a matrix Sylvester equation. The spectrum of $\hat{\bm{\Lambda}}$ and the spectrum of  $\mathbf{A} - \mathbf{P}$ are disjoint with high probability and hence Theorem~VII.2.1 and Theorem~VII.2.2 in \cite{bhatia} implies
\begin{equation}
\label{eq:Uhat_Sylvester}
 \hat{\mathbf{U}} = \sum_{k=0}^{\infty} (\mathbf{A} - \mathbf{P})^{k} \mathbf{P} \hat{\mathbf{U}} \hat{\bm{\Lambda}}^{-(k+1)} = 
\sum_{k=0}^{\infty} (\mathbf{A} - \mathbf{P})^{k} \mathbf{U} \bm{\Lambda} \bU^{\top} \hat{\mathbf{U}} \hat{\bm{\Lambda}}^{-(k+1)}.
 \end{equation}
 with high probability. Eq.~\eqref{eq:Uhat_Sylvester} also implies
\begin{equation}
\label{eq:Uhat_Sylvester2}
 \bm{\Pi}^{\perp}_{\bU} \hat{\mathbf{U}} = 
\bm{\Pi}^{\perp}_{\bU} \sum_{k=1}^{\infty} (\mathbf{A} - \mathbf{P})^{k} \mathbf{U} \bm{\Lambda} \bU^{\top} \hat{\mathbf{U}} \hat{\bm{\Lambda}}^{-(k+1)}.
 \end{equation}
 Several key steps in our proof of Theorem~\ref{THM:GEN_D} and Theorem~\ref{THM:GEN_D_SPARSE} proceed by using Lemma~\ref{lem:A-P} to truncate the series expansions in Eq.~\eqref{eq:Uhat_Sylvester} and Eq.~\eqref{eq:Uhat_Sylvester2}. More specifically, we have the following result.
\begin{lemma}
\label{lem:von_Neuman_hatU}
Let $\mathbf{A} \sim \mathrm{GRDPG}_{p,q}(F)$ be a generalized random dot product graph on $n$ vertices with sparsity factor $\rho_n$. Then with $\mathbf{E} = \mathbf{A} - \mathbf{P}$, we have
\begin{equation}
\label{eq:approximation_transpose1}
\begin{split}
\bU^{\top} \hat{\bU} \hat{\bm{\Lambda}} - \bm{\Lambda} \bU^{\top} \hat{\bU} & = \bU^{\top} \mathbf{A} \hat{\bU} - \bU^{\top} \bP \hat{\bU} = \bU^{\top} \Bigl(\sum_{k=1}^{\infty} \mathbf{E}^{k} \bU \bm{\Lambda} \bU^{\top} \hat{\bU} \hat{\bm{\Lambda}}^{-k} \Bigr) \\
&= \bU^{\top} \mathbf{E} \bU \bm{\Lambda} \bU^{\top} \hat{\bU} \hat{\bm{\Lambda}}^{-1} + \bU^{\top} \mathbf{E}^{2} \bU \bm{\Lambda} \bU^{\top} \hat{\bU} \hat{\bm{\Lambda}}^{-2} + O_{\mathbb{P}}((n \rho_n)^{-1/2}) \\
&= \bU^{\top} \mathbf{E} \bU \bU^{\top} \hat{\bU} + \bU^{\top} \mathbf{E}^{2} \bU \bm{\Lambda}^{-1} \bU^{\top} \hat{\bU} + O_{\mathbb{P}}((n \rho_n)^{-1/2})
\\ &= O_{\mathbb{P}}(1).
\end{split}
\end{equation}
\begin{gather}
\label{eq:approximate_transpose2}
\bU^{\top} \hat{\bU} \hat{\bm{\Lambda}}^{-1} - \bm{\Lambda}^{-1} \bU^{\top} \hat{\bU} = O_{\mathbb{P}}((n \rho_n)^{-2}) \\
\label{eq:approximate_transpose3}
\bU^{\top} \hat{\bU} \hat{\bm{\Lambda}}^{-2} - \bm{\Lambda}^{-2} \bU^{\top} \hat{\bU} = O_{\mathbb{P}}((n \rho_n)^{-3}). 
\end{gather}
In addition, we also have
\begin{equation}
\begin{split}
\label{eq:I-UUT_Uhat1}
\bm{\Pi}_{\bU}^{\perp} \hat{\bU} & = \bm{\Pi}^{\perp}_{\bU} \mathbf{E} \bU \bm{\Lambda} \bU^{\top} \hat{\bU} \hat{\bm{\Lambda}}^{-2} + O_{\mathbb{P}}((n \rho_n)^{-1}) \\
&= \bm{\Pi}^{\perp}_{\bU} \mathbf{E} \bU \bm{\Lambda}^{-1} \bU^{\top} \hat{\bU} + O_{\mathbb{P}}((n \rho_n)^{-1}) \\
&= \mathbf{E} \mathbf{U} \bm{\Lambda}^{-1} \mathbf{U}^{\top} \hat{\bU} + O_{\mathbb{P}}((n \rho_n)^{-1}),
\end{split}
\end{equation}
\begin{equation}
\label{eq:I-UUT_Uhat2}
\begin{split}
\bm{\Pi}_{\bU}^{\perp} \hat{\bU} \hat{\bm{\Lambda}} & = \bm{\Pi}^{\perp}_{\bU} \mathbf{E} \bigl( \bU \bm{\Lambda} \bU^{\top} \hat{\bU} \hat{\bm{\Lambda}}^{-1} + \mathbf{E} \bU \bm{\Lambda} \bU^{\top} \hat{\bU} \hat{\bm{\Lambda}}^{-2} \bigr) + 
O_{\mathbb{P}}((n \rho_n)^{-1/2}) \\
&= \bm{\Pi}^{\perp}_{\bU} \mathbf{E} \bm{\Pi}_{\bU} \hat{\bU} + \bm{\Pi}^{\perp}_{\bU} \mathbf{E}^{2} \bU \bm{\Lambda}^{-1} \bU^{\top} \hat{\bU} + O_{\mathbb{P}}((n \rho_n)^{-1/2}).
\end{split}
\end{equation}
\end{lemma}
\begin{proof}
We first derive parts of Eq.~\eqref{eq:approximation_transpose1}. From Lemma~\ref{lem:A-P}, we obtain
\begin{equation*}
\begin{split}
 \Bigl \| \sum_{k=3}^{\infty} \mathbf{E}^{k} \bU \bm{\Lambda} \bU^{\top} \hat{\bm{\Lambda}}^{-k} \Bigr \| & \leq \sum_{k=3}^{\infty} \|\mathbf{E}^{k} \| \times \|\bm{\Lambda} \| \times
 \|\hat{\bm{\Lambda}}^{-k}\| \\ & \leq \sum_{k=3}^{\infty} O_{\mathbb{P}} (n \rho_n)^{-(k-1)/2}) = O_{\mathbb{P}}((n \rho_n)^{-1/2}).
 \end{split}
 \end{equation*}
 and hence
 \begin{equation} 
 \label{eq:approximation_transpose1b}
 \begin{split} \bU^{\top} \hat{\bU} \hat{\bm{\Lambda}} - \bm{\Lambda} \bU^{\top} \hat{\bU} = \bU^{\top} \mathbf{E} \bU \bm{\Lambda} \bU^{\top} \hat{\bU} \hat{\bm{\Lambda}}^{-1} & + \bU^{\top} \mathbf{E}^{2} \bU \bm{\Lambda} \bU^{\top} \hat{\bU} \hat{\bm{\Lambda}}^{-2} \\ &+ O_{\mathbb{P}}((n \rho_n)^{-1/2}). 
 \end{split}
 \end{equation}
 Let $\bm{u}_i$ denote the $i$-th column of $\bU$. 
 We note that $\bU^{\top} \mathbf{E} \bU$ is a $d \times d$ matrix whose $ij$-th entry can be written as $\bm{u}_i^{\top} \mathbf{E} \bm{u}_j$. Now, conditioned on $\mathbf{P}$, $\bm{u}_i^{\top} \mathbf{E} \bm{u}_j$ is a sum of independent mean $0$ random variables, and hence, by Hoeffding's inequality, $\bm{u}_i^{\top} \mathbf{E} \bm{u}_j = O_{\mathbb{P}}(1)$. A union bound over the $d(d+1)/2$ upper triangular entries of $\bU^{\top} \mathbf{E} \bU$ then yield
 $ \|\bU^{\top} \mathbf{E} \bU \| = O_{\mathbb{P}}(1)$.
We therefore have
 \begin{equation*}
 \begin{split}
 \|\bU^{\top} \hat{\bU} \hat{\bm{\Lambda}} - \bm{\Lambda} \bU^{\top} \hat{\bU} \| & \leq \|\bU^{\top} \mathbf{E} \bU \| \times \|\bm{\Lambda} \| \times \|\hat{\bm{\Lambda}}\|^{-1} + \|\mathbf{E}^{2} \| \times \|\bm{\Lambda} \| \times \|\hat{\bm{\Lambda}}\|^{-2} \\ &+ O_{\mathbb{P}}((n \rho_n)^{-1/2}) 
 = O_{\mathbb{P}}(1).
 \end{split}
 \end{equation*}
 We next show Eq.~\eqref{eq:approximate_transpose2}. Let $\omega_{ij}$ denote the $ij$-th entry of $\mathbf{U}^{\top} \hat{\mathbf{U}}$ and let $\hat{\lambda}_i$ and $\lambda_i$ denote the $i$-th diagonal element of $\hat{\bm{\Lambda}}$ and $\bm{\Lambda}$ respectively, i.e., $\hat{\lambda}_i$ and $\lambda_i$ are the $i$-th largest eigenvalue, in modulus, of $\mathbf{A}$ and $\mathbf{P}$. Then the $ij$- th entry of $\mathbf{U}^{\top} \hat{\mathbf{U}} \hat{\bm{\Lambda}}^{-1} - \bm{\Lambda}^{-1} \mathbf{U}^{\top} \hat{\mathbf{U}}$ can be written as
 $$ \omega_{ij} \bigl(\hat{\lambda}_j^{-1} - \lambda_i^{-1}) = \omega_{ij} \frac{\lambda_i - \hat{\lambda}_j}{\lambda_i \hat{\lambda}_j}.$$
 Therefore, letting $\mathbf{H}$ denote the $d \times d$ matrix whose $ij$-th entry is $\lambda_i^{-1} \hat{\lambda}_j^{-1}$, we have (with $\circ$ denoting the Hadamard product between matrices)
 $$ \mathbf{U}^{\top} \hat{\mathbf{U}} \hat{\bm{\Lambda}}^{-1} - \bm{\Lambda}^{-1} \mathbf{U}^{\top} \hat{\mathbf{U}} = (\bm{\Lambda} \bU^{\top} \hat{\bU} - \bU^{\top} \hat{\bU} \hat{\bm{\Lambda}}) \circ \mathbf{H} = O_{\mathbb{P}}((n \rho_n)^{-2}).$$
 Eq.~\eqref{eq:approximate_transpose3} is derived in an analogous manner. More specifically, let $\tilde{\mathbf{H}}$ denote the $d \times d$ matrix whose $ij$-th entry is $\hat{\lambda}_j^{-2} \lambda_i^{-2} (\hat{\lambda}_j + \lambda_i)$, we have
\begin{equation*}
 \mathbf{U}^{\top} \hat{\mathbf{U}} \hat{\bm{\Lambda}}^{-2} - \bm{\Lambda}^{-2} \mathbf{U}^{\top} \hat{\mathbf{U}} = (\bm{\Lambda} \bU^{\top} \hat{\bU} - \bU^{\top} \hat{\bU} \hat{\bm{\Lambda}}) \circ \tilde{\mathbf{H}} = O_{\mathbb{P}}((n \rho_n)^{-3}).
\end{equation*}
We then apply Eq.~\eqref{eq:approximate_transpose2} and Eq.~\eqref{eq:approximate_transpose3} to Eq.~\eqref{eq:approximation_transpose1b} and obtain another representation for $\bU^{\top} \hat{\bU} \hat{\bm{\Lambda}} - \bm{\Lambda} \bU^{\top} \hat{\bU}$, namely
$$ \bU^{\top} \hat{\bU} \hat{\bm{\Lambda}} - \bm{\Lambda} \bU^{\top} \hat{\bU} = \bU^{\top} \mathbf{E} \bU \bU^{\top} \hat{\bU} + \bU^{\top} \mathbf{E}^{2} \bU \bm{\Lambda}^{-1} \bU^{\top} \hat{\bU} + O_{\mathbb{P}}((n \rho_n)^{-1/2}).$$
Eq.~\eqref{eq:approximation_transpose1} is thereby established. Eq.~\eqref{eq:I-UUT_Uhat1} and Eq.~\eqref{eq:I-UUT_Uhat2} is derived in a similar manner to that of Eq.~\eqref{eq:approximation_transpose1}. 
\end{proof}

\subsection*{Deriving Eq.~\eqref{eq:xi_kl3} and Eq.~\eqref{eq:xi_kl2}}
 We start with the observation
\begin{equation*}
\begin{split}
\bU^{\top} \hat{\bU} \hat{\bU}^{\top} \bU - \mathbf{I} &= \bU^{\top} \hat{\bU} \mathbf{W}^{\top} \mathbf{W} \hat{\bU}^{\top} \bU - \mathbf{I} \\
&= - (\bU - \hat{\bU} \mathbf{W}^{\top})^{\top} (\bU - \hat{\bU} \mathbf{W}) + \mathbf{U}^{\top} (\bU - \hat{\bU} \mathbf{W}) (\bU - \hat{\bU} \mathbf{W})^{\top} \bU.
\end{split}
\end{equation*}
Now $\|\mathbf{U}^{\top} (\bU - \hat{\bU} \mathbf{W})\| = \|\mathbf{I} - \bm{\Sigma}\|$ where $\bm{\Sigma}$ is the diagonal matrix whose diagonal entries are the singular values of $\bU^{\top} \hat{\bU}$. Lemma~\ref{lem:A-P} then implies $\|\mathbf{U}^{\top} (\bU - \hat{\bU} \mathbf{W})\| = \mathbb{O}_{\mathbb{P}}((n \rho_n)^{-1})$ and hence 
\begin{equation}
\label{eq:form1}
\bU^{\top} \hat{\bU} \hat{\bU}^{\top} \bU - \mathbf{I} = - 
(\bU - \hat{\bU} \mathbf{W}^{\top})^{\top} (\bU - \hat{\bU} \mathbf{W}) + O_{\mathbb{P}}((n \rho_n)^{-2}).
\end{equation}
We recall the following bounds
\begin{gather}
\|s_k\| = \sqrt{n_k} = \Theta(\sqrt{n}); \quad \|s_{\ell}\| = \sqrt{n_{\ell}} = 
\Theta(\sqrt{n}) \\ n \rho_n = \omega(\sqrt{n}); \quad \|\bm{\Lambda}\| = O_{\mathbb{P}}(n \rho_n).
\end{gather}
Eq.~\eqref{eq:form1} and Lemma~\ref{lem:von_Neuman_hatU} then imply
\begin{equation*}
\label{eq:xi_kl}
\begin{split}
\xi_{k\ell}^{(3)} &= - \tfrac{n \rho_n^{-1/2}}{n_{k} n_{\ell}} \bm{s}_k^{\top} \bU \bLam (\bU - \hat{\bU} \mathbf{W}^{\top})^{\top} (\bU - \hat{\bU} \mathbf{W}) \bU^{\top} \bm{s}_{\ell} + O_{\mathbb{P}}(n^{-1} \rho_n^{-3/2}) \\
&= -\tfrac{n \rho_n^{-1/2}}{n_{k} n_{\ell}} \bm{s}_k^{\top} \bU \bU^{\top} (\mathbf{A} - \mathbf{P}) (\bU - \hat{\bU} \mathbf{W}) \bU^{\top} \bm{s}_{\ell} + O_{\mathbb{P}}(n^{-1/2} \rho_n) + O_{\mathbb{P}}(n^{-1} \rho_n^{-3/2}) \\
&= -\tfrac{n \rho_n^{-1/2}}{n_{k} n_{\ell}} \bm{s}_k^{\top} \bU \bU^{\top} (\mathbf{A} - \mathbf{P})^{2} \bU \bLam^{-1} \bU^{\top} \bm{s}_{\ell} + O_{\mathbb{P}}(n^{-1/2} \rho_n)
\end{split}
\end{equation*}
thereby establishing Eq.~\eqref{eq:xi_kl3}.

We next derive Eq.~\eqref{eq:xi_kl2}. We recall Eq.~\eqref{eq:approximation_transpose1} in Lemma~\ref{lem:von_Neuman_hatU}, namely that
 \begin{equation*}
 \begin{split} \mathbf{U}^{\top} \hat{\mathbf{U}} \hat{\bm{\Lambda}} - \bm{\Lambda} \mathbf{U}^{\top} \hat{\mathbf{U}} 
 &= \mathbf{U}^{\top} \mathbf{E} \mathbf{U} \mathbf{U}^{\top} \hat{\mathbf{U}}  + \mathbf{U}^{\top} \mathbf{E}^{2} \mathbf{U} \bm{\Lambda}^{-1} \mathbf{U}^{\top} \hat{\mathbf{U}} + O_{\mathbb{P}}((n \rho_n)^{-1/2}).
\end{split}
\end{equation*}
We therefore have, in conjunction with Eq.~\eqref{eq:form1}, that
$$ (\mathbf{U}^{\top} \hat{\mathbf{U}} \hat{\bm{\Lambda}} - \bm{\Lambda} \mathbf{U}^{\top} 
\hat{\mathbf{U}}) \hat{\bU}^{\top} \bU = \mathbf{U}^{\top} \mathbf{E} \mathbf{U} + \mathbf{U}^{\top} \mathbf{E}^2 \mathbf{U} \bm{\Lambda}^{-1} + O_{\mathbb{P}}((n \rho_n)^{-1/2}).$$
and hence
\begin{equation*}
\begin{split}
\xi_{k \ell}^{(2)} &= \tfrac{n \rho_n^{-1/2}}{n_{k} n_{\ell}} (\bm{s}_k^{\top} \bU (\bU^{\top} \hat{\bU} \hat{\bLam} - \bLam \bU^{\top} \hat{\bU}) 
\hat{\bU}^{\top} \bU \bU^{\top} \bm{s}_{\ell}) \\
&= \tfrac{n \rho_n^{-1/2}}{n_{k} n_{\ell}} s_k^{\top} \bU \Bigl(\mathbf{U}^{\top} \mathbf{E} \mathbf{U} + \mathbf{U}^{\top} \mathbf{E}^2 \mathbf{U} \bm{\Lambda}^{-1} + O_{\mathbb{P}}((n \rho_n)^{-1/2}) \Bigr) \bU^{\top} \bm{s}_{\ell} \\
&= \tfrac{n \rho_n^{-1/2}}{n_{k} n_{\ell}} s_k^{\top} \bigl(\bm{\Pi}_{\bU} \mathbf{E} \bm{\Pi}_{\bU} + \bm{\Pi}_{\bU} \mathbf{E}^2 \mathbf{P}^{\dagger} \bigr) \bm{s}_{\ell} + O_{\mathbb{P}}(n^{-1/2} \rho_n^{-1}) 
\\ &= \tfrac{n \rho_n^{-1/2}}{n_{k} n_{\ell}} s_k^{\top} \bm{\Pi}_{\bU} \mathbf{E} \bm{\Pi}_{\bU} \bm{s}_{\ell} - \xi_{kl}^{(3)} + O_{\mathbb{P}}(n^{-1/2} \rho_n^{-1}). 
\end{split}
\end{equation*}
as desired. 

\subsection*{Deriving Eq.~\eqref{eq:xi_kl1}}
We start with the decomposition
\begin{equation*}
\begin{split}
\hat{\mathbf{U}} \hat{\bm{\Lambda}} \hat{\mathbf{U}}^{\top} - \mathbf{U} \mathbf{U}^{\top} \hat{\mathbf{U}}
\hat{\bLam} \hat{\bU}^{\top} \bU \bU^{\top} & = \hat{\mathbf{U}} \hat{\bm{\Lambda}} \hat{\mathbf{U}}^{\top} - \bm{\Pi}_{\bU} \hat{\bU} \hat{\bLam} \hat{\bU}^{\top} \bm{\Pi}_{\bU} \\
&= \bm{\Pi}_{\bU}^{\perp} \hat{\bU} \hat{\bm{\Lambda}} \hat{\mathbf{U}}^{\top} \bm{\Pi}_{\bU}^{\perp} + 
\bm{\Pi}_{\bU}^{\perp} \hat{\bU} \hat{\bm{\Lambda}} \hat{\mathbf{U}}^{\top} \bm{\Pi}_{\bU} + \bm{\Pi}_{\bU} \hat{\bU} \hat{\bm{\Lambda}} \hat{\mathbf{U}}^{\top} \bm{\Pi}_{\bU}^{\perp}.
\end{split}
\end{equation*}
Now let $\omega_{kl}^{(1)} = \tfrac{n \rho_n^{-1/2}}{n_{k} n_{\ell}} \bm{s}_k^{\top} \bm{\Pi}_{\bU}^{\perp} \hat{\bU} \hat{\bm{\Lambda}} \hat{\bU}^{\top} \bm{\Pi}^{\perp} \bm{s}_{\ell}$. 
By Eq.~\eqref{eq:I-UUT_Uhat2} in Lemma~\ref{lem:von_Neuman_hatU}, we have
\begin{equation*}
\begin{split} 
\omega_{kl}^{(1)} &=
\tfrac{n \rho_n^{-1/2}}{n_{k} n_{\ell}} \bm{s}_k^{\top} \bm{\Pi}_{\bU}^{\perp} \hat{\bU} \hat{\bm{\Lambda}} \hat{\bU}^{\top} \bm{\Pi}_{\bU}^{\perp} \bm{s}_{\ell} \\
& = \tfrac{n \rho_n^{-1/2}}{n_{k} n_{\ell}} \bm{s}_k^{\top} \bigl(\bm{\Pi}^{\perp}_{\bU} \mathbf{E} \bm{\Pi}_{\bU} \hat{\bU} + \bm{\Pi}^{\perp}_{\bU} \mathbf{E}^{2} \bU \bm{\Lambda}^{-1} \bU^{\top} \hat{\bU} + O_{\mathbb{P}}((n \rho_n)^{-1/2})\bigr) \hat{\bU}^{\top} \bm{\Pi}_{\bU}^{\perp} \bm{s}_{\ell}
\end{split}
\end{equation*}
From Lemma~\ref{lem:A-P}, we have $\|\bm{\Pi}_{\bU}^{\perp} \hat{\bU}\| = O_{\mathbb{P}}((n \rho_n)^{-1/2})$, and hence
$$
\omega_{kl}^{(1)} = \tfrac{n \rho_n^{-1/2}}{n_{k} n_{\ell}}  \bm{s}_k^{\top} \bm{\Pi}_{\bU}^{\perp} 
\mathbf{E} \bU \bLam \bU^{\top} \hat{\bU} \hat{\bLam}^{-1} \hat{\bU}^{\top} \bm{\Pi}_{\bU}^{\perp} \bm{s}_{\ell} + O_{\mathbb{P}}(n^{-1/2} \rho_n^{-1})
$$
Now, conditional on $\mathbf{P}$, $\bm{s}_k^{\top} \bm{\Pi}_{\bU}^{\perp} 
(\bA - \bP) \bU$ is vector in $\mathbb{R}^{d}$ whose elements are sum of independent mean $0$ random variables. Therefore, by Hoeffding's inequality and the fact that $\|\bm{s}_k\| = \Theta(\sqrt{n})$ and $\|\bU\|_{F} = \sqrt{d}$, we have
$$ \bm{s}_k^{\top} \bm{\Pi}_{\bU}^{\perp} (\bA - \bP) \bU = O_{\mathbb{P}}(\sqrt{n}) $$
and thus
\begin{equation}
\label{eq:perp_perp} 
\begin{split}
|\omega_{kl}^{(1)}| &= 
|\tfrac{n \rho_n^{-1/2}}{n_{k} n_{\ell}} \bm{s}_k^{\top} \bm{\Pi}_{\bU}^{\perp} \hat{\bU} \hat{\bm{\Lambda}} \hat{\bU}^{\top} \bm{\Pi}^{\perp}_{\bU} \bm{s}_{\ell}|
\\& \leq \rho_n^{-1/2} \times \|\bLam\| \times \|\hat{\bLam}^{-1} \| \times \|\hat{\bU}^{\top} \bm{\Pi}_{\bU}^{\perp} \| + O_{\mathbb{P}}(n^{-1/2} \rho_n^{-1}) \\
&= O_{\mathbb{P}}(n^{-1/2} \rho_n^{-1}).
\end{split}
\end{equation}
Next let $\omega_{kl}^{(2)} := \tfrac{n \rho_n^{-1/2}}{n_{k} n_{\ell}} \bm{s}_k^{\top} \bm{\Pi}_{\bU}^{\perp} \hat{\bU} \hat{\bm{\Lambda}} \hat{\mathbf{U}}^{\top} \bm{\Pi}_{\bU} \bm{s}_{\ell}$. Once again Eq.~\eqref{eq:I-UUT_Uhat2} implies
\begin{equation*}
\begin{split}
 \omega_{kl}^{(2)} 
 &= \tfrac{n \rho_n^{-1/2}}{n_{k} n_{\ell}} \bm{s}_k^{\top} \bm{\Pi}_{\bU}^{\perp} \bigl(\mathbf{E} \bU \bLam \bU^{\top} \hat{\bm{U}}  \hat{\bLam}^{-1} + 
\mathbf{E}^{2} \bU \bLam \bU^{\top} \hat{\bm{U}}  \hat{\bLam}^{-2}  + O_{\mathbb{P}}((n \rho_n)^{-1/2}) \bigr) \hat{\bU} ^{\top} \bm{\Pi}_{\bU} \bm{s}_{\ell} \\
&=  
\tfrac{n \rho_n^{-1/2}}{n_{k} n_{\ell}} \bm{s}_k^{\top} \bm{\Pi}_{\bU}^{\perp} \bigl(\mathbf{E} \bU \bLam \bU^{\top} \hat{\bm{U}}  \hat{\bLam}^{-1} + 
\mathbf{E}^{2} \bU \bLam \bU^{\top} \hat{\bm{U}}  \hat{\bLam}^{-2}) \hat{\bU}^{\top} \bm{\Pi}_{\bU} \bm{s}_{\ell} + O_{\mathbb{P}}((n \rho_n)^{-1/2})
 \end{split}
\end{equation*}
Applying Eq.~\eqref{eq:approximate_transpose2} and Eq.~\eqref{eq:approximate_transpose3} to the above yield
\begin{equation*}
\omega_{kl}^{(2)} = \tfrac{n \rho_n^{-1/2}}{n_{k} n_{\ell}} \bm{s}_k^{\top} \bigl(\bm{\Pi}_{\bU}^{\perp} \mathbf{E} \bU \bU^{\top} \hat{\bm{U}} \hat{\bm{U}}^{\top} \bm{\Pi}_{\bU} + \bm{\Pi}_{\bU}^{\perp}
\mathbf{E}^{2} \bU \bLam^{-1} \bU^{\top} \hat{\bm{U}} \hat{\bU}^{\top} \bm{\Pi}_{\bU} \bigr) \bm{s}_{\ell} + O_{\mathbb{P}}((n \rho_n)^{-1/2}).
\end{equation*}
Using Eq.~\eqref{eq:form1}, we replace $\bU \bU^{\top} \hat{\bm{U}} \hat{\bm{U}}^{\top}$ by $\mathbf{I}$ and replace $\bU^{\top} \hat{\bm{U}} \hat{\bU}^{\top} \bm{\Pi}_{\bU} = \bU^{\top} \hat{\bm{U}} \hat{\bU}^{\top} \bU \bU^{\top}$ by $\bU^{\top}$ in the above display, thereby obtaining
$$ \omega_{kl}^{(2)} = \tfrac{n \rho_n^{-1/2}}{n_{k} n_{\ell}} \bm{s}_k^{\top} \bigl(\bm{\Pi}_{\bU}^{\perp} \mathbf{E} \bm{\Pi}_{\bm{U}}   + \bm{\Pi}_{\bU}^{\perp}
\mathbf{E}^{2} \bU \bLam^{-1} \bU^{\top} \bigr) \bm{s}_{\ell} + O_{\mathbb{P}}((n \rho_n)^{-1/2}).
$$
By exchanging $k$ and $\ell$, we also have
\begin{equation*}
\begin{split} \omega_{kl}^{(3)} &:= \tfrac{n \rho_n^{-1/2}}{n_{k} n_{\ell}} \bm{s}_{\ell}^{\top} \bm{\Pi}_{\bU}^{\perp} \hat{\bU} \hat{\bm{\Lambda}} \hat{\mathbf{U}}^{\top} \bm{\Pi}_{\bU} \bm{s}_{k} \\ &=  \tfrac{n \rho_n^{-1/2}}{n_{k} n_{\ell}} \bm{s}_{\ell}^{\top} (\bm{\Pi}_{\bU}^{\perp} \mathbf{E} \bm{\Pi}_{\bU} + \bm{\Pi}_{\bU}^{\perp} \mathbf{E}^2 \bU \bLam^{-1} \bU^{\top} \bigr) \bm{s}_{k} + O_{\mathbb{P}}(n^{-1/2} \rho_n^{-1})
\end{split}
\end{equation*}
Combining the above expressions for $\omega_{kl}^{(1)}, \omega_{kl}^{(2)}$ and $\omega_{kl}^{(3)}$, we obtain
\begin{equation*} 
\begin{split}
\xi_{kl}^{(3)} &= \omega_{kl}^{(1)} + \omega_{kl}^{(2)} + \omega_{kl}^{(3)} 
\\ &= \tfrac{n \rho_n^{-1/2}}{n_{k} n_{\ell}} \bigl(\bm{s}_k^{\top} 
\bm{\Pi}_{\bU}^{\perp} \mathbf{E} \bm{\Pi}_{\bU} \bm{s}_{\ell} + \bm{s}_{k}^{\top} \bm{\Pi}_{\bU}^{\perp} \mathbf{E}^2 \mathbf{P}^{\dagger} \bm{s}_{\ell} + 
\bm{s}_{\ell}^{\top} 
\bm{\Pi}_{\bU}^{\perp} \mathbf{E} \bm{\Pi}_{\bU} \bm{s}_{k} + \bm{s}_{\ell}^{\top} \bm{\Pi}_{\bU}^{\perp} \mathbf{E}^2 \mathbf{P}^{\dagger} \bm{s}_{k} \bigr) 
\\ &+ 
O_{\mathbb{P}}(n^{-1/2} \rho_n^{-1})
\end{split}
\end{equation*}
as desired. 

\subsection*{Deriving Eq.~\eqref{eq:sigma_kk}, Eq.~\eqref{eq:sigma_kl2}, Eq.~\eqref{eq:tilde_sigma_kk} and Eq.~\eqref{eq:tilde_sigma_kl}}
We first recall Eq.~\eqref{eq:Zkl},
$$ Z_{kl} = \tfrac{n \rho_n^{-1/2}}{n_{k} n_{\ell}} \mathrm{tr} \,(\mathbf{A} - \mathbf{P}) (\bm{\Pi}_{\bU} \bm{s}_{\ell} \bm{s}_k^{\top} + \bm{\Pi}_{\bU} \bm{s}_{k} \bm{s}_{\ell}^{\top}  - \bm{\Pi}_{\bU} \bm{s}_{k} \bm{s}_{\ell}^{\top} \bm{\Pi}_{\bU}). $$
With $\mathbf{M} = \bm{\Pi}_{\bU} \bm{s}_{\ell} \bm{s}_k^{\top} + \bm{\Pi}_{\bU} \bm{s}_{k} \bm{s}_{\ell}^{\top}  - \bm{\Pi}_{\bU} \bm{s}_{k} \bm{s}_{\ell}^{\top} \bm{\Pi}_{\bU}$, we have
\begin{equation}
\begin{split}
Z_{kl} &= \tfrac{n \rho_n^{-1/2}}{n_{k} n_{\ell}} \sum_{i} \sum_{j} (\mathbf{A}_{ij} - \mathbf{P}_{ij}) \mathbf{M}_{ij} \\
&= \tfrac{n \rho_n^{-1/2}}{n_{k} n_{\ell}} \Bigl(\sum_{i < j} (\mathbf{A}_{ij} - \mathbf{P}_{ij}) (\bM_{ij} + \bM_{ji}) + \sum_{i} (\mathbf{A}_{ii} - \mathbf{P}_{ii}) \bM_{ii} \Bigr) \\
\end{split}
\end{equation}
which is a sum of $n(n+1)/2$ independent mean $0$ random variables. By the Lindeberg-Feller central limit theorem, $Z_{kl} \overset{\mathrm{d}}{\longrightarrow} N(0, \mathrm{Var}[Z_{kl}])$. 
All that remains is to evaluate $\mathrm{Var}[Z_{kl}]$.

Let $\mathbf{X}$ be the $n \times d$ matrix such that $X_i$, the $i$-th row of $\mathbf{X}$, is $\nu_k$ if $\tau_i = k$, i.e., $\mathbf{X} \mathbf{I}_{p,q} \mathbf{X}^{\top} = \mathbb{E}[\mathbf{A}]$. We observe that $\bm{\Pi}_{\bU} = \bU \bU^{\top} = \mathbf{X} (\mathbf{X}^{\top} \mathbf{X})^{-1} \mathbf{X}^{\top}$ as $\bm{\Pi}_{\bU}$ is the orthogonal projection onto the column space of $\mathbf{X} \mathbf{I}_{p,q} \mathbf{X}^{\top}$ which coincides with that of $\mathbf{X}$. 
Let $\tau = (\tau_1, \dots, \tau_n)$ be the vertices to block assignments of $\mathbf{A}$. Then the $ij$-th entries of $\bm{\Pi}_{\bU} \bm{s}_{\ell} \bm{s}_k^{\top}$, $\bm{\Pi}_{\bU} \bm{s}_{\ell} \bm{s}_k^{\top}$, and $\bm{\Pi}_{\bU} \bm{s}_{k} \bm{s}_{\ell}^{\top} \bm{\Pi}_{\bU}$ are
\begin{gather*}
(\bm{\Pi}_{\bU} \bm{s}_{\ell} \bm{s}_k^{\top})_{ij} = n_{\ell} X_i^{\top} (\mathbf{X}^{\top} \mathbf{X})^{-1} \nu_{\ell} \ast \mathbbm{1}\{\tau_j = k\}, \\
(\bm{\Pi}_{\bU} \bm{s}_{k} \bm{s}_{\ell}^{\top})_{ij} = n_k X_i^{\top} (\mathbf{X}^{\top} \mathbf{X})^{-1} \nu_{k} \ast \mathbbm{1}\{\tau_j = \ell\},\\
(\bm{\Pi}_{\bU} \bm{s}_{k} \bm{s}_{\ell}^{\top} \bm{\Pi}_{\bU})_{ij} = n_{k} n_{\ell} X_i^{\top} (\mathbf{X}^{\top} \mathbf{X})^{-1} \nu_{k} \nu_l^{\top} (\mathbf{X}^{\top} \mathbf{X})^{-1} X_{j},
\end{gather*}
and hence
\begin{equation}
\label{eq:M_ij+M_ji}
\begin{split}
\mathbf{M}_{ij} + \mathbf{M}_{ji} &= n_{k} \nu_{\ell}^{\top}(\mathbf{X}^{\top} \mathbf{X})^{-1} \bigl(X_i \mathbbm{1}\{\tau_j = k\} + X_j \mathbbm{1}\{\tau_i = k \}\bigr)
\\
&+ n_{\ell} \nu_k^{\top} (\mathbf{X}^{\top} \mathbf{X})^{-1} \bigl( X_i \mathbbm{1}\{\tau_j = \ell\} + X_j \mathbbm{1}\{\tau_i = \ell\} \bigr)
\\& - n_{k} n_{\ell} X_i^{\top} (\mathbf{X}^{\top} \mathbf{X})^{-1} \bigl(\nu_{\ell} \nu_k^{\top} + \nu_{k} \nu_{\ell}^{\top}\bigr) (\mathbf{X}^{\top} \mathbf{X})^{-1} X_{j}.
\end{split}
\end{equation}
Next, we note that
\begin{equation*}
\begin{split}
\mathrm{Var}[Z_{kl}] &= \tfrac{n^{2} \rho_n^{-1}}{n_k^2 n_{\ell}^2} \sum_{i < j} \mathbf{P}_{ij} (1 - \mathbf{P}_{ij}) (\bM_{ij} + \bM_{ji})^{2} + \tfrac{n^{2} \rho_n^{-1}}{n_k^2 n_{\ell}^2} \sum_{i} \mathbf{P}_{ii} (1 - \mathbf{P}_{ii} )\mathbf{M}_{ii}^2 \\
&= \tfrac{n^{2} \rho_n^{-1}}{2 n_k^2 n_{\ell}^2} \sum_{i} \sum_{j} \mathbf{P}_{ij} (1 - \mathbf{P}_{ij}) (\bM_{ij} + \bM_{ji})^{2} + o_{\mathbb{P}}(1) \\
&= S_{kk} + 2 S_{k\ell}  + S_{\ell\ell} + 2 S_{ko} + 2 S_{\ell o} + S_{oo} + o_{\mathbb{P}}(1)
\end{split}
\end{equation*}
where each $S_{\ast\ast}$ correspond to summing $\eta_{ij} := \mathbf{P}_{ij} (1 - \mathbf{P}_{ij}) (\bM_{ij} + \bM_{ji})^{2}$ 
over some subset of the indices $(i,j)$, namely
 \begin{gather*} 
 S_{kk} = \tfrac{n^{2} \rho_n^{-1}}{2 n_{k}^{2} n_{\ell}^{2}} \sum_{\tau_i = k} \sum_{\tau_j = k} \eta_{ij}, \,\, S_{\ell\ell} =
 \tfrac{n^{2} \rho_n^{-1}}{2 n_{k}^{2} n_{\ell}^{2}}  \sum_{\tau_i = \ell} \sum_{\tau_j = \ell} 
 \eta_{ij}, \\ 
 S_{k\ell} = \tfrac{n^{2} \rho_n^{-1}}{2 n_{k}^{2} n_{\ell}^{2}}  \sum_{\tau_i = k} \sum_{\tau_j = \ell} \eta_{ij}, \,\,
 S_{ko} = \tfrac{n^{2} \rho_n^{-1}}{2 n_{k}^{2} n_{\ell}^{2}}  \sum_{\tau_i = k} \sum_{\tau_j \not \in \{k,\ell\}} \eta_{ij}, \\
S_{\ell o} = \tfrac{n^{2} \rho_n^{-1}}{2 n_{k}^{2} n_{\ell}^{2}}  \sum_{\tau_i = \ell} \sum_{\tau_j \not \in \{k,\ell\}} \eta_{ij}, \,\,
S_{oo} = \tfrac{n^{2} \rho_n^{-1}}{2 n_{k}^{2} n_{\ell}^{2}} \sum_{\tau_i \not \in \{k,\ell\}} \sum_{\tau_j \not \in \{k,\ell\}} \eta_{ij}.
 \end{gather*}
If $k \not = \ell$, then for $(i,j)$ such that $\tau_i = k$ and $\tau_j = k$, Eq.~\eqref{eq:M_ij+M_ji} yield
 \begin{equation*} 
\begin{split}
 \mathbf{M}_{ij} + \mathbf{M}_{ji} &= 2 n_{\ell} \nu_k^{\top} (\mathbf{X}^{\top} \mathbf{X})^{-1} \nu_{\ell} - 
2 n_{k} n_{\ell} \nu_k^{\top} (\mathbf{X}^{\top} \mathbf{X})^{-1} \bm{\nu}_{\ell} \nu_k^{\top} (\mathbf{X}^{\top} \mathbf{X})^{-1} \nu_k.
\end{split}
\end{equation*}
and hence, since $\mathbf{P}_{ij} = \rho_n \mathbf{B}_{\tau_i,\tau_j}$, 
\begin{equation*}
\begin{split}
S_{kk} &= \tfrac{2n^2 \rho_n^{-1}}{n_k^{2} n_{\ell}^2} n_k^2 \rho_n \mathbf{B}_{kk} (1 - \rho_n \mathbf{B}_{kk}) (n_\ell \nu_k^{\top} (\mathbf{X}^{\top} \mathbf{X})^{-1} \nu_{\ell} - 
n_k n_\ell \nu_k^{\top} (\mathbf{X}^{\top} \mathbf{X})^{-1} \bm{\nu}_{\ell} \nu_k^{\top} (\mathbf{X}^{\top} \mathbf{X})^{-1} \nu_k)^2 \\
&= 2 \mathbf{B}_{kk}(1 - \rho_n \mathbf{B}_{kk})(\bm{\nu}_k^{\top} (\tfrac{\mathbf{X}^{\top} \mathbf{X}}{n})^{-1} \bm{\nu}_{\ell} - \tfrac{n_k}{n} \nu_k^{\top} (\tfrac{\mathbf{X}^{\top} \mathbf{X}}{n})^{-1} \nu_{\ell} \bm{\nu}_k^{\top} (\tfrac{\mathbf{X}^{\top} \mathbf{X}}{n})^{-1} \bm{\nu}_{k})^{2}
\end{split}
\end{equation*}
We therefore have
\begin{equation*} S_{kk} \overset{\mathrm{a.s.}}{\longrightarrow} \begin{cases} 2 \mathbf{B}_{kk}(1 - \mathbf{B}_{kk}) \zeta_{k \ell}^{2} (1 - \pi_k \zeta_{kk})^2 & \text{if $\rho_n \equiv 1$ for all $n$} \\
2 \mathbf{B}_{kk} \zeta_{k \ell}^{2} (1 - \pi_k \zeta_{kk})^2 & \text{if $\rho_n \rightarrow 0$}
\end{cases}
\end{equation*}
as $n \rightarrow \infty$.
Similarly, we have
\begin{equation*}
S_{\ell \ell} \overset{\mathrm{a.s.}}{\longrightarrow} \begin{cases} 2 \mathbf{B}_{\ell \ell}(1 - \mathbf{B}_{\ell \ell}) \zeta_{k \ell}^{2} (1 - \pi_{\ell} \zeta_{\ell \ell})^2 & \text{if $\rho_n \equiv 1$ for all $n$} \\
2 \mathbf{B}_{\ell \ell} \zeta_{k \ell}^{2} (1 - \pi_{\ell} \zeta_{\ell \ell})^2 & \text{if $\rho_n \rightarrow 0$}
\end{cases}
\end{equation*}
as $n \rightarrow \infty$.
If $k \not = \ell$, then for $(i,j)$ with $\tau_i = k$ and $\tau_j = \ell$, Eq.~\eqref{eq:M_ij+M_ji} yield
\begin{equation*} 
\begin{split} \mathbf{M}_{ij} + \mathbf{M}_{ji} &= n_{\ell} \nu_{\ell}^{\top} (\mathbf{X}^{\top} \mathbf{X})^{-1} \nu_{\ell} + n_k \nu_{k}^{\top} (\mathbf{X}^{\top} \mathbf{X})^{-1} \nu_{k} \\ & - n_{k} n_{\ell} \nu_k^{\top} (\mathbf{X}^{\top} \mathbf{X})^{-1} \nu_{k} \nu_{\ell}^{\top} (\mathbf{X}^{\top} \mathbf{X})^{-1} 
\nu_{\ell} - n_{k} n_{\ell} (\nu_k^{\top} (\mathbf{X}^{\top} \mathbf{X})^{-1} \nu_{\ell})^2
\end{split}
\end{equation*}
and hence
\begin{equation*}
\begin{split}
S_{k\ell}  \overset{\mathrm{a.s.}}{\longrightarrow} \begin{cases} \tfrac{1}{2}\pi_{k} \pi_{\ell} \mathbf{B}_{k\ell}(1 - \mathbf{B}_{k\ell})  \bigl(\tfrac{1}{\pi_k} \zeta_{\ell \ell} + 
\tfrac{1}{\pi_{\ell}} \zeta_{kk} - \zeta_{kk} \zeta_{\ell \ell} - \zeta_{k \ell}^2\bigr)^2 & \text{if $\rho_n \equiv 1$} \\
\tfrac{1}{2}\pi_{k} \pi_{\ell} \mathbf{B}_{k\ell} \bigl(\tfrac{1}{\pi_k} \zeta_{\ell \ell} + 
\tfrac{1}{\pi_{\ell}} \zeta_{kk} - \zeta_{kk} \zeta_{\ell \ell} - \zeta_{k \ell}^2\bigr)^2 & \text{if $\rho_n \rightarrow 0$}
\end{cases}
\end{split}
\end{equation*}
as $n \rightarrow \infty$.
If $k \not = \ell$ then for $(i,j)$ with $\tau_i = k, \tau_j \not \in \{k,\ell\}$, Eq.~\eqref{eq:M_ij+M_ji} yield
\begin{equation*} 
\begin{split} \mathbf{M}_{ij} + \mathbf{M}_{ji} &= n_{\ell} X_j^{\top} (\mathbf{X}^{\top} \mathbf{X})^{-1} \nu_{\ell} - n_{k} n_{\ell} \nu_k^{\top} (\mathbf{X}^{\top} \mathbf{X})^{-1} \nu_{k} \nu_{\ell}^{\top} (\mathbf{X}^{\top} \mathbf{X})^{-1} 
X_j \\ & - n_{k} n_{\ell} \nu_k^{\top} (\mathbf{X}^{\top} \mathbf{X})^{-1} \nu_{\ell} \nu_{k}^{\top} (\mathbf{X}^{\top} \mathbf{X})^{-1} 
X_j
\end{split}
\end{equation*}
and hence
\begin{equation*}
S_{ko} \overset{\mathrm{a.s.}}{\longrightarrow} \begin{cases} \tfrac{1}{2} \sum_{r \not \in \{k,\ell\}} \pi_{k} \pi_{r} \mathbf{B}_{k r}(1 - \mathbf{B}_{kr})  (\tfrac{1}{\pi_k} \zeta_{\ell r} - \zeta_{kk} \zeta_{\ell r} - \zeta_{k \ell} \zeta_{kr})^2 & \text{if $\rho_n \equiv 1$} \\
\tfrac{1}{2} \sum_{r \not \in \{k,\ell\}} \pi_{k} \pi_{r} \mathbf{B}_{k r} (\tfrac{1}{\pi_k} \zeta_{\ell r} - \zeta_{kk} \zeta_{\ell r} - \zeta_{k \ell} \zeta_{kr})^2 & \text{if $\rho_n \rightarrow 0$}
\end{cases}
\end{equation*}
as $n \rightarrow \infty$. By symmetry, we also have
\begin{equation*}
S_{ko} \overset{\mathrm{a.s.}}{\longrightarrow} \begin{cases} \tfrac{1}{2} \sum_{r \not \in \{k,\ell\}} \pi_{\ell} \pi_{r} \mathbf{B}_{\ell r}(1 - \mathbf{B}_{\ell r})  (\tfrac{1}{\pi_\ell} \zeta_{k r} - \zeta_{\ell \ell} \zeta_{k r} - \zeta_{k \ell} \zeta_{\ell r})^2 & \text{if $\rho_n \equiv 1$} \\
\tfrac{1}{2} \sum_{r \not \in \{k,\ell\}} \pi_{\ell} \pi_{r} \mathbf{B}_{\ell r} (\tfrac{1}{\pi_\ell} \zeta_{k r} - \zeta_{\ell \ell} \zeta_{k r} - \zeta_{k \ell} \zeta_{\ell r})^2 & \text{if $\rho_n \rightarrow 0$}
\end{cases}
\end{equation*}
Finally, when $(i,j)$ is such that $\tau_i \not \in \{k,\ell\}$ and $\tau_j \not \in \{k,\ell\}$, Eq.~\eqref{eq:M_ij+M_ji} yield
$$
\mathbf{M}_{ij} + \mathbf{M}_{ji} = - n_{k} n_{\ell} X_i^{\top} (\mathbf{X}^{\top} \mathbf{X})^{-1} (\nu_{k} \nu_{\ell}^{\top} + \nu_{\ell} \nu_k^{\top}) (\mathbf{X}^{\top} \mathbf{X})^{-1} X_{j} 
$$
and thus
$$ S_{oo} \overset{\mathrm{a.s.}}{\longrightarrow} \begin{cases}  \tfrac{1}{2} \sum_{r \not \in \{k,\ell\}} \sum_{s \not \in \{k,\ell\}} \pi_r \pi_s 
\mathbf{B}_{rs} (1 - \mathbf{B}_{rs}) (\zeta_{kr} \zeta_{\ell s} + \zeta_{\ell r} \zeta_{ks})^2 & \text{if $\rho_n \equiv 1$} \\
\tfrac{1}{2} \sum_{r \not \in \{k,\ell\}} \sum_{s \not \in \{k,\ell\}} \pi_r \pi_s 
\mathbf{B}_{rs} (\zeta_{kr} \zeta_{\ell s} + \zeta_{\ell r} \zeta_{ks})^2 & \text{if $\rho_n \rightarrow 0$}
\end{cases}
 $$
 as $n \rightarrow \infty$.
Combining the above expressions for $S_{kk}, S_{k\ell}, S_{\ell\ell}, S_{\ell o}, S_{k o}$ and $S_{oo}$ yield $\sigma^{2}_{k \ell}$ and $\tilde{\sigma}^{2}_{k \ell}$. For example, with $\rho_n \equiv 1$,
\begin{equation*}
\begin{split}
\sigma_{k\ell}^{2} &= 2 \mathbf{B}_{kk}(1 - \mathbf{B}_{kk}) \zeta_{k \ell}^{2} (1 - \pi_{k} \zeta_{kk})^{2} + 2 \mathbf{B}_{\ell}(1 - \mathbf{B}_{\ell \ell}) \zeta_{k \ell}^{2} (1 - \pi_{\ell} \zeta_{\ell \ell})^{2} \\ &+ \pi_{k} \pi_{\ell} \mathbf{B}_{k\ell}(1 - \mathbf{B}_{k\ell})  \bigl(\tfrac{1}{\pi_k} \zeta_{\ell \ell} + 
\tfrac{1}{\pi_{\ell}} \zeta_{kk} - \zeta_{kk} \zeta_{\ell \ell} - \zeta_{k \ell}^2\bigr)^2 \\
&+ \sum_{r \not \in \{k,\ell\}} \pi_{k} \pi_{r} \mathbf{B}_{k r}(1 - \mathbf{B}_{kr})  (\tfrac{1}{\pi_k} \zeta_{\ell r} - \zeta_{kk} \zeta_{\ell r} - \zeta_{k \ell} \zeta_{kr})^2 \\ &+
\sum_{r \not \in \{k,\ell\}} \pi_{k} \pi_{r} \mathbf{B}_{k r}(1 - \mathbf{B}_{kr})  (\tfrac{1}{\pi_k} \zeta_{\ell r} - \zeta_{kk} \zeta_{\ell r} - \zeta_{k \ell} \zeta_{kr})^2
\\ &+ \tfrac{1}{2} \sum_{r \not \in \{k,\ell\}} \sum_{s \not \in \{k,\ell\}} \pi_r \pi_s 
\mathbf{B}_{rs} (1 - \mathbf{B}_{rs}) (\zeta_{kr} \zeta_{\ell s} + \zeta_{\ell r} \zeta_{ks})^2
\end{split}
\end{equation*}
for when $k \not = \ell$. Straightforward manipulations then yield the form given in Eq.~\eqref{eq:sigma_kl2}. 

When $k = \ell$, the term $\mathrm{Var}[Z_{kk}]$ is decomposed as $\mathrm{Var}[Z_{kk}] = S_{kk} + 2 S_{ko} + S_{oo}$ where we now have
\begin{equation*}
S_{kk} = \tfrac{n^{2} \rho_n^{-1}}{2 n_{k}^{4}} \sum_{\tau_i = k} \sum_{\tau_j = k} \eta_{ij}, \,\, 
 S_{ko} = \tfrac{n^{2} \rho_n^{-1}}{2 n_{k}^{4}}  \sum_{\tau_i = k} \sum_{\tau_j \not = k} \eta_{ij}, \,\,
S_{oo} = \tfrac{n^{2} \rho_n^{-1}}{2 n_{k}^{4}} \sum_{\tau_i \not = k} \sum_{\tau_j \not = k} \eta_{ij}.
\end{equation*} 
If $k = \ell$, then for $(i,j)$ such that $\tau_{i} = k = \ell$ and $\tau_j = k = \ell$, Eq.~\eqref{eq:M_ij+M_ji} yield
$$ \mathbf{M}_{ij} + \mathbf{M}_{ji} = 4 n_k \nu_k^{\top} (\mathbf{X}^{\top} \mathbf{X})^{-1} \nu_k - 2 n_k^2 (\nu_k^{\top} (\mathbf{X}^{\top} \mathbf{X})^{-1} \nu_k)^2 $$
from which we obtain
\begin{equation*}
 \mathbf{S}_{kk} \overset{\mathrm{a.s.}}{\longrightarrow} \begin{cases}
2 \mathbf{B}_{kk}(1 - \mathbf{B}_{kk}) \zeta_{k k}^{2} (2 - \pi_k \zeta_{kk})^2 & \text{if $\rho_n \equiv 1$ for all $n$} \\
2 \mathbf{B}_{kk} \zeta_{k k}^{2} (2 - \pi_k \zeta_{kk})^2 & \text{if $\rho_n \rightarrow 0$}.
\end{cases}
 \end{equation*}
 If $k = \ell$, then for $(i,j)$ such that $\tau_{i} = k$ and $\tau_j \not = k$, Eq.~\eqref{eq:M_ij+M_ji} yield
 $$ \mathbf{M}_{ij} + \mathbf{M}_{ji} = 2 n_k \nu_k^{\top} (\mathbf{X}^{\top} \mathbf{X})^{-1} X_j - 2 n_k^2 \nu_k^{\top} (\mathbf{X}^{\top} \mathbf{X})^{-1} \nu_k \nu_k^{\top} (\mathbf{X}^{\top} \mathbf{X})^{-1} X_j $$
 and thus
 \begin{equation*}
 \mathbf{S}_{ko} \overset{\mathrm{a.s.}}{\longrightarrow} \begin{cases}
2 \sum_{r \not = k} \pi_k \pi_r \mathbf{B}_{kr}(1 - \mathbf{B}_{kr}) \zeta_{kr}^{2} (\tfrac{1}{\pi_k} - \zeta_{kk})^2 & \text{if $\rho_n \equiv 1$ for all $n$} \\
2 \sum_{r \not = k} \pi_k \pi_r \mathbf{B}_{kr} \zeta_{kr}^{2} (\tfrac{1}{\pi_k} - \zeta_{kk})^2 & \text{if $\rho_n \rightarrow 0$}.
\end{cases}
 \end{equation*}
 Finally, for $k = \ell$ and $\tau_i \not = k$, $\tau_j \not = k$, we have
 \begin{gather*} \mathbf{M}_{ij} + \mathbf{M}_{ji} = - 2 n_k^2 X_i (\mathbf{X}^{\top} \mathbf{X})^{-1} \nu_k \nu_k^{\top} (\mathbf{X}^{\top} \mathbf{X})^{-1} X_j \\ 
 S_{oo} = \overset{\mathrm{a.s.}}{\longrightarrow} \begin{cases}  2 \sum_{r \not = k} \sum_{s \not = k} \pi_r \pi_s 
\mathbf{B}_{rs} (1 - \mathbf{B}_{rs}) \zeta_{kr}^{2} \zeta_{k s}^2 & \text{if $\rho_n \equiv 1$} \\
2 \sum_{r \not = k} \sum_{s \not = k} \pi_r \pi_s 
\mathbf{B}_{rs} \zeta_{kr}^{2} \zeta_{k s}^2 & \text{if $\rho_n \rightarrow 0$}.
\end{cases}
 \end{gather*}
 Combining the above expressions for $S_{kk}, S_{ko}$ and $S_{oo}$ and some straightforward manipulations yield us Eq.~\eqref{eq:sigma_kk} and Eq.~\eqref{eq:tilde_sigma_kk}. 

\subsection*{Deriving Eq.~\eqref{eq:mu_kl} and Eq.~\eqref{eq:tilde_theta_kl}}
Our argument is similar to that of \cite{tang14:_semipar} and is based on a log-Sobolev concentration inequality of \cite{boucheron2003} which yield
\begin{gather*} \tfrac{n}{n_k n_{\ell}} \bm{s}_{k}^{\top} \bm{\Pi}_{\bU}^{\perp} (\mathbf{A} - \mathbf{P})^{2} \bU \bLam^{-1} \bU^{\top} \bm{s}_{\ell} = 
\mathbb{E}[\tfrac{n}{n_{k} n_{\ell}} \bm{s}_{k}^{\top} \bm{\Pi}_{\bU}^{\perp} (\mathbf{A} - \mathbf{P})^{2} \bU \bLam^{-1} \bU^{\top} \bm{s}_{\ell}] + O_{\mathbb{P}}(n^{-1/2}) 
\\ \tfrac{n}{n_{k} n_{\ell}} \bm{s}_{\ell}^{\top} \bm{\Pi}_{\bU}^{\perp} (\mathbf{A} - \mathbf{P})^{2} \bU \bLam^{-1} \bU^{\top} \bm{s}_{k} = \mathbb{E}[\tfrac{n}{n_k n_{\ell}} \bm{s}_{\ell}^{\top} \bm{\Pi}_{\bU}^{\perp} (\mathbf{A} - \mathbf{P})^{2} \bU \bLam^{-1} \bU^{\top} \bm{s}_{k}] + O_{\mathbb{P}}(n^{-1/2})
\end{gather*}
where the expectations are taken with respect to $\mathbf{A}$, conditional on $\mathbf{P}$. We now evaluate $\theta_{k\ell}^{(1)} := \tfrac{n}{n_{k} n_{\ell}} \mathbb{E}[\bm{s}_{k}^{\top} \bm{\Pi}_{\bU}^{\perp} (\mathbf{A} - \mathbf{P})^{2} \bU \bLam^{-1} \bU^{\top} \bm{s}_{\ell}]$. Let $\mathbf{D} = \mathbb{E}[(\mathbf{A} - \mathbf{P})]^2$ be the diagonal matrix whose diagonal entries are
$$ \mathbf{D}_{ii} = \sum_{j} \mathbf{P}_{ij} (1 - \mathbf{P}_{ij}) = \rho_n \sum_{r=1}^{K} n_{r} X_i^{\top} \mathbf{I}_{p,q} \nu_r (1 - \rho_n X_i^{\top} \mathbf{I}_{p,q} \nu_r).$$
Next, we note $\mathbf{U} \bm{\Lambda} \mathbf{U}^{\top} = \mathbf{P} = \rho_n \mathbf{X} \mathbf{I}_{p,q} \mathbf{X}^{\top}$ and $\mathbf{U} \bm{\Lambda}^{-1} \mathbf{U}^{\top}$ is the Moore-Penrose pseudo-inverse $\mathbf{P}^{\dagger}$ of $\mathbf{P}$. 
Since the Moore-Penrose pseudoinverse of $\mathbf{P}$ is unique, we therefore have
$$\mathbf{U} \bm{\Lambda}^{-1} \mathbf{U}^{\top} = \mathbf{P}^{\dagger} = (\rho_n \mathbf{X} \mathbf{I}_{p,q} \mathbf{X}^{\top})^{\dagger} = \rho_n^{-1} \mathbf{X} (\mathbf{X}^{\top} \mathbf{X})^{-1} \mathbf{I}_{p,q} (\mathbf{X}^{\top} \mathbf{X})^{-1} \mathbf{X}^{\top}.$$
We then have
\begin{equation*}
\begin{split}
\theta_{k\ell}^{(1)} &= \tfrac{n}{n_{k} n_{\ell}} \bm{s}_k^{\top} \bm{\Pi}_{\mathbf{U}}^{\perp} \mathbb{E}[(\mathbf{A} - \mathbf{P})^2] \mathbf{U} \bm{\Lambda}^{-1} \mathbf{U}^{\top} \bm{s}_{\ell} \\
&= \tfrac{n}{\rho_n n_{k} n_{\ell}} \bm{s}_k^{\top} \bm{\Pi}_{\mathbf{U}}^{\perp} \mathbf{D} \mathbf{X} (\mathbf{X}^{\top} \mathbf{X})^{-1} \mathbf{I}_{p,q} (\mathbf{X}^{\top} \mathbf{X})^{-1} \mathbf{X}^{\top} \bm{s}_{\ell} \\
&= \tfrac{n}{\rho_n n_k n_{\ell}} \bm{s}_k^{\top} (\mathbf{I} - \mathbf{X} (\mathbf{X}^{\top} \mathbf{X})^{-1} \mathbf{X}^{\top}) \mathbf{D} \mathbf{X} (\mathbf{X}^{\top} \mathbf{X})^{-1} \mathbf{I}_{p,q} (\mathbf{X}^{\top} \mathbf{X})^{-1} \mathbf{X}^{\top} \bm{s}_{\ell} \\
&= \tfrac{n}{\rho_n n_k} \Bigl(\bm{s}_{k}^{\top} \mathbf{D} \mathbf{X}
(\mathbf{X}^{\top} \mathbf{X})^{-1} \mathbf{I}_{p,q} (\mathbf{X}^{\top} \mathbf{X})^{-1} \nu_{\ell} - n_k \nu_k^{\top} (\mathbf{X}^{\top} \mathbf{X})^{-1} \mathbf{X}^{\top} \mathbf{D} \mathbf{X} (\mathbf{X}^{\top} \mathbf{X})^{-1} \mathbf{I}_{p,q} (\mathbf{X}^{\top} \mathbf{X})^{-1} \nu_{\ell} \Bigr)
\end{split}
\end{equation*}
Letting $\zeta_{kl}^{(1)} = \tfrac{n}{\rho_n n_k} \bm{s}_{k}^{\top} \mathbf{D} \mathbf{X}
(\mathbf{X}^{\top} \mathbf{X})^{-1} \mathbf{I}_{p,q} (\mathbf{X}^{\top} \mathbf{X})^{-1} \nu_{\ell}$, some straightforward simplifications yield
\begin{equation*}
\begin{split}
\zeta_{kl}^{(1)} &= 
\frac{n}{\rho_n n_k} \bm{s}_k^{\top} \mathbf{D} \mathbf{X} (\bX^{\top} \bX)^{-1} \mathbf{I}_{p,q} (\mathbf{X}^{\top} \mathbf{X})^{-1} \nu_{\ell} \\ &= n \sum_{r = 1}^{K} n_r \nu_k^{\top} \mathbf{I}_{p,q} \nu_r (1 - \rho_n \nu_k^{\top} \mathbf{I}_{p,q} \nu_r) \nu_k^{\top} (\bX^{\top} \bX)^{-1} \mathbf{I}_{p,q} (\mathbf{X}^{\top} \mathbf{X})^{-1} \nu_{\ell} \\
&\overset{\mathrm{a.s.}}{\longrightarrow} \begin{cases} \sum_{r=1}^{K} \pi_r \mathbf{B}_{kr} (1 - \mathbf{B}_{kr}) \nu_k^{\top} \Delta^{-1} \mathbf{I}_{p,q} \Delta^{-1} \nu_{\ell} & \text{if $\rho_n \equiv 1$} \\
\sum_{r=1}^{K} \pi_r \mathbf{B}_{kr}  \nu_k^{\top} \Delta^{-1} \mathbf{I}_{p,q} \Delta^{-1} \nu_{\ell} & \text{if $\rho_n \rightarrow 0$}
\end{cases}
\end{split}
\end{equation*}
as $n \rightarrow \infty$. Similarly, 
letting $\zeta_{k \ell}^{(2)} = n \rho_n^{-1} \nu_k^{\top} (\bX^{\top} \bX)^{-1} \bX^{\top} \mathbf{D} \bX (\bX^{\top} \bX)^{-1} \mathbf{I}_{p,q} (\mathbf{X}^{\top} \mathbf{X})^{-1} \nu_{\ell}$,
\begin{equation*}
\begin{split}
\zeta_{kl}^{(2)} 
&= n \rho_n^{-1} \sum_{i} \nu_k^{\top} (\bX^{\top} \bX)^{-1} X_i \mathbf{D}_{ii} X_i^{\top} (\bX^{\top} \bX)^{-2} \nu_{\ell} \\
&= n \sum_{s=1}^{K} n_s \nu_k^{\top} (\bX^{\top} \bX)^{-1} \nu_s \sum_{r=1}^{K} n_r \nu_s^{\top} \mathbf{I}_{p,q} \nu_r (1 - \rho_n \nu_s^{\top} \mathbf{I}_{p,q} \nu_r) \nu_s^{\top} (\bX^{\top} \bX)^{-1} \mathbf{I}_{p,q} (\bX^{\top} \bX)^{-1} 
\nu_{\ell} \\ 
&\overset{\mathrm{a.s.}}{\longrightarrow} \begin{cases} \sum_{s=1}^{K} \sum_{r=1}^{K} \pi_r \pi_s \nu_k^{\top} \Delta^{-1} \nu_s \mathbf{B}_{sr} (1 - \mathbf{B}_{sr}) \nu_s^{\top} \Delta^{-1} \mathbf{I}_{p,q} \Delta^{-1} \nu_{\ell} & \text{if $\rho_n \equiv 1$} \\
 \sum_{s=1}^{K} \sum_{r=1}^{K} \pi_r \pi_s \nu_k^{\top} \Delta^{-1} \nu_s \mathbf{B}_{sr} \nu_s^{\top} \Delta^{-1} \mathbf{I}_{p,q} \Delta^{-1} \nu_{\ell} & \text{if $\rho_n \rightarrow 0$}
\end{cases}
\end{split}
\end{equation*}
as $n \rightarrow \infty$. Now $\theta_{k \ell}^{(1)} = \zeta_{k \ell}^{(1)} + \zeta_{k \ell}^{(2)}$. We therefore have, for $\rho_n \equiv 1$, that
\begin{equation*}
\begin{split}
 \theta_{k\ell}^{(1)} &= \sum_{r=1}^{K} \pi_r \mathbf{B}_{kr} (1 - \mathbf{B}_{kr}) \nu_k^{\top} \Delta^{-1} \mathbf{I}_{p,q}  \Delta^{-1} \nu_{\ell} \\ &- 
\sum_{r=1}^{K} \sum_{s=1}^{K} 
\pi_r \pi_s \mathbf{B}_{sr} (1 - \mathbf{B}_{sr}) \nu_s^{\top} \Delta^{-1} \mathbf{I}_{p,q} \Delta^{-1} \nu_{\ell} \nu_k^{\top} \Delta^{-1} \nu_s.
\end{split}
\end{equation*}
In contrast, if $\rho_n \rightarrow 0$, then
\begin{equation*}
\begin{split}
 \theta_{k\ell}^{(1)} = \sum_{r=1}^{K} \pi_r \mathbf{B}_{kr} \nu_k^{\top} \Delta^{-1} \mathbf{I}_{p,q}  \Delta^{-1} \nu_{\ell} - 
\sum_{r=1}^{K} \sum_{s=1}^{K} 
\pi_r \pi_s \mathbf{B}_{sr} \nu_s^{\top} \Delta^{-1} \mathbf{I}_{p,q} \Delta^{-1} \nu_{\ell} \nu_k^{\top} \Delta^{-1} \nu_s.
\end{split}
\end{equation*}
Swapping $\ell$ with $k$ in the above expression yield a similar expression for $\theta_{k\ell}^{(2)} = \mathbb{E}[n \bm{s}_{\ell}^{\top} \bm{\Pi}_{\bU}^{\perp} (\mathbf{A} - \mathbf{P})^{2} \bU \bLam^{-1} \bU^{\top} \bm{s}_{k}]$. Since $\theta_{k \ell}^{(1)} + \theta_{k \ell}^{(2)} = \theta_{k \ell}$ for $\rho_n \equiv 1$, we conclude 
\begin{equation*}
\begin{split}
\theta_{k\ell} &= \sum_{r=1}^{K} \pi_r \bigl(\mathbf{B}_{kr} (1 - \mathbf{B}_{kr}) + \mathbf{B}_{\ell r}(1 - \mathbf{B}_{\ell r})\bigr)\nu_k^{\top} \Delta^{-1} \mathbf{I}_{p,q} \Delta^{-1} \nu_{\ell} \\
& - \sum_{r=1}^{K} \sum_{s=1}^{K} \pi_r \pi_s \mathbf{B}_{sr} (1 - \mathbf{B}_{sr}) \nu_s^{\top} \Delta^{-1} \mathbf{I}_{p,q} \Delta^{-1} (\nu_{\ell} \nu_{k}^{\top} + \nu_{k} \nu_{\ell}^{\top}) \Delta^{-1} \nu_s.
\end{split}
\end{equation*} 
Similarly, $\theta_{k \ell}^{(1)} + \theta_{k \ell}^{(2)} = \tilde{\theta}_{k \ell}$ when $\rho_n \rightarrow 0$, and hence
\begin{equation*}
\begin{split}
\tilde{\theta}_{k\ell} &= \sum_{r=1}^{K} \pi_r \bigl(\mathbf{B}_{kr}  + \mathbf{B}_{\ell r} \bigr)\nu_k^{\top} \Delta^{-1} \mathbf{I}_{p,q} \Delta^{-1} \nu_{\ell} \\
& - \sum_{r=1}^{K} \sum_{s=1}^{K} \pi_r \pi_s \mathbf{B}_{sr} 
\nu_s^{\top} \Delta^{-1} \mathbf{I}_{p,q} \Delta^{-1} (\nu_{\ell} \nu_{k}^{\top} + \nu_{k} \nu_{\ell}^{\top}) \Delta^{-1} \nu_s
\end{split}
\end{equation*} 
as desired.

\subsection*{Proof of Lemma~\ref{LEM:PERFECT}}
We recall the notion of the $2 \to \infty$ norm for matrices, namely, for a $n \times m$ matrix $\mathbf{A}$ (with $\mathbf{A}_i$ denoting the $i$-th row of $\mathbf{A}$)
$$ \|\mathbf{A} \|_{2 \to \infty} = \max_{\|\bm{x}\|_2 = 1} \|\mathbf{A} \bm{x}\|_{\infty} = \max_{i \in [n]} \|\mathbf{A}_i\|_2.
$$
Eq.~\eqref{eq:perfect} in Lemma~\ref{LEM:PERFECT} can thus be rewritten as
\begin{equation}
\label{eq:perfect2} \|\hat{\mathbf{U}}_n - \mathbf{U}_n \mathbf{W} \|_{2 \to \infty} = O_{\mathbb{P}}\Bigl(\frac{\log^{c}{n}}{n \sqrt{\rho_n}}\Bigr)
\end{equation}
for some orthogonal $\mathbf{W}$. We now derive Eq.~\eqref{eq:perfect2}. For ease of exposition, we shall drop the index $n$ from our matrices $\mathbf{X}_n$, $\mathbf{A}_n$, $\hat{\mathbf{U}}_n$ and $\mathbf{U}_n$. We first note that for any matrices $\mathbf{A}$ and $\mathbf{B}$ whose product $\mathbf{A} \mathbf{B}$ is well-defined, 
$\|\mathbf{A} \mathbf{B}\|_{2 \to \infty} \leq \|\mathbf{A}\|_{2 \to \infty} \times \|\mathbf{B}\|$. Next, we note that $\|\mathbf{U} \|_{2 \to \infty} = O_{\mathbb{P}}(n^{-1/2})$ as the rows of $\mathbf{X} $ are sampled i.i.d. from $F$. Recalling Lemma~\ref{lem:A-P}, we then have
\begin{equation*} 
\begin{split} \|\hat{\mathbf{U}} - \mathbf{U} \mathbf{W} \|_{2 \to \infty} & \leq
\|\hat{\mathbf{U}} - \mathbf{U} \mathbf{U} ^{\top} \hat{\mathbf{U}} \|_{2 \to \infty} + \|\mathbf{U} \mathbf{U}^{\top} \hat{\mathbf{U}} - \mathbf{U} \mathbf{W}\|_{2 \to \infty} \\ 
&\leq \|\hat{\mathbf{U}} - \mathbf{U} \mathbf{U}^{\top} \hat{\mathbf{U}} \|_{2 \to \infty} + \|\mathbf{U}\|_{2 \to \infty} \|\mathbf{U}^{\top} \hat{\mathbf{U}} - \mathbf{W} \| \\
&\leq \|\hat{\mathbf{U}} - \mathbf{U} \mathbf{U}^{\top} \hat{\mathbf{U}} \|_{2 \to \infty} + O_{\mathbb{P}}\Bigl(\tfrac{1}{n^{3/2} \rho_n}\Bigr).
\end{split}
\end{equation*}
 Eq.~\eqref{eq:Uhat_Sylvester2} now implies (recall that $\bm{\Pi}_{\bU}^{\perp} = \mathbf{I} - \bU \bU^{\top}$)
 \begin{equation}
 \label{eq:lem_perfect_part0}
 \begin{split} \|\hat{\mathbf{U}} - \mathbf{U} \mathbf{U}^{\top} \hat{\mathbf{U}} \|_{2 \to \infty} & \leq \sum_{k=1}^{\infty} \|\bm{\Pi}_{\bU}^{\perp} (\mathbf{A} - \mathbf{P})^{k} \bU \bm{\Lambda} \bU^{\top} \hat{\bU} \hat{\bm{\Lambda}}^{-(k+1)} \|_{2 \to \infty} \\
 &\leq \sum_{k=1}^{\infty} \|(\mathbf{A} - \mathbf{P})^{k} \bU\|_{2 \to \infty} \times \|\bm{\Lambda}\| \times \|\hat{\bm{\Lambda}}^{-1}\|^{(k+1)} \\ &+ \sum_{k=1}^{\infty} \|\bU\|_{2 \to \infty} \times \|\bU^{\top} 
 (\mathbf{A} - \mathbf{P})^{k} \bU \| \times \|\bm{\Lambda}\| \times \|\hat{\bm{\Lambda}}^{-1}\|^{(k+1)}.
\end{split}
 \end{equation}
 Once again, by Lemma~\ref{lem:A-P}, we have
 \begin{equation}
 \label{eq:lem_perfect_part1}
 \begin{split} \sum_{k=1}^{\infty} \|\bU\|_{2 \to \infty} \times \|\bU^{\top} 
 (\mathbf{A} - \mathbf{P})^{k} \bU \| \times \|\bm{\Lambda}\| \times \|\hat{\bm{\Lambda}}^{-1}\|^{(k+1)} & \leq \sum_{k=1}^{\infty} \mathbb{O}_{\mathbb{P}}\Bigl(\frac{1}{\sqrt{n} (n \rho_n)^{k/2}}\Bigr) \\ &= \mathbb{O}_{\mathbb{P}}\Bigl(\frac{1}{n \sqrt{\rho_n}}\Bigr).
 \end{split}
\end{equation}
We now bound $\sum_{k=1}^{\infty} \|(\mathbf{A} - \mathbf{P})^{k} \bU\|_{2 \to \infty} \times \|\bm{\Lambda}\| \times \|\hat{\bm{\Lambda}}^{-1}\|^{(k+1)}$. We need the following slight restatement of Lemma~7.10 from \cite{erdos}. 
\begin{lemma}
\label{lem:erdos}
Assume the setting and notations in Lemma~\ref{LEM:PERFECT}. Let $\bm{u}_j$ be the $j$-th column of $\mathbf{U}$ for $j = 1,2, \dots, d$. Then there exists constants $c  > 0$ such that for all $k \leq \log{n}$ 
$$ \|(\mathbf{A} - \mathbf{P})^{k} \bU\|_{2 \to \infty} \leq
\sqrt{d} \max_{j \in [d]} \|(\mathbf{A} - \mathbf{P})^{k} \bm{u}_j\|_{\infty} = O_{\mathbb{P}}\Bigl(\frac{\sqrt{d} (n \rho_n)^{k/2} \log^{kc}(n)}{\sqrt{n}}\Bigr).$$
\end{lemma}
We note that Lemma~7.10 from \cite{erdos} was originally
stated for the case when 
$\bm{u}_j = n^{-1/2} \bm{1}$ \footnote{There is a small typo in \cite{erdos} in that for Lemma~7.10, $\bm{e}$ is defined as $\bm{e} = \bm{1}$, while $\bm{e} = n^{-1/2} \bm{1}$ is used everywhere else in the paper.}, but the argument used in the proof of Lemma~7.10 can be easily extended to the setting where the entries of $\bm{u}_j$ are ``delocalized'', i.e., $\|\bm{u}_j\|_{\infty} = O_{\mathbb{P}}(n^{-1/2})$. Using Lemma~\ref{lem:erdos} and Lemma~\ref{LEM:PERFECT}, we obtain
\begin{equation*}
\begin{split}
\sum_{k=1}^{\infty} \|(\mathbf{A} - \mathbf{P})^{k} \bU\|_{2 \to \infty} \|\bm{\Lambda}\| \|\hat{\bm{\Lambda}}^{-1}\|^{(k+1)} &
\leq \sum_{k=1}^{\log{n}} O_{\mathbb{P}}\Bigl(\tfrac{\sqrt{d} \log^{kc}(n)}{\sqrt{n} (n \rho_n)^{k/2}}\Bigr) + \sum_{k > \log{n}} O_{\mathbb{P}}((n \rho_n)^{-k/2}) \\
& \leq O_{\mathbb{P}}\Bigl(\tfrac{\sqrt{d} \log^{c}(n)}{n \sqrt{\rho_n}}\Bigr) + O_{\mathbb{P}}((n \rho_n)^{-(\tfrac{1}{2} \log{n})}).
\end{split}
\end{equation*}
If we now assume $n \rho_n = \omega(\log^{2c}(n))$, then $$(n \rho_n)^{-(\tfrac{1}{2} \log{n})}) = o_{\mathbb{P}}\Bigl(\tfrac{\sqrt{d} \log^{c}(n)}{n \sqrt{\rho_n}}\Bigr)$$ and hence
\begin{equation}
\label{eq:lem_perfect_part2}
\begin{split}
\sum_{k=1}^{\infty} \|(\mathbf{A} - \mathbf{P})^{k} \bU\|_{2 \to \infty} \times \|\bm{\Lambda}\| \times \|\hat{\bm{\Lambda}}^{-1}\|^{(k+1)} \leq O_{\mathbb{P}}\Bigl(\tfrac{\sqrt{d} \log^{c}(n)}{n \sqrt{\rho_n}}\Bigr).
\end{split}
\end{equation}
Substituting Eq.~\eqref{eq:lem_perfect_part1} and Eq.~\eqref{eq:lem_perfect_part2} into Eq.~\eqref{eq:lem_perfect_part0} yield
Eq.~\eqref{eq:perfect2}, as desired.
\bibliography{../biblio}
\end{document}